\documentclass[11pt]{article}

\def\final{1}

\usepackage{amsmath,amsfonts, mathdots, amssymb, latexsym, amsbsy,bm,amsthm,mathtools}       

\usepackage{braket}
\usepackage[margin=1in]{geometry}
\usepackage[english]{babel}
\usepackage[utf8]{inputenc}
\usepackage{xcolor}

\colorlet{myblue}{blue!80!black}
\colorlet{mygreen}{green!40!black}

\usepackage{xspace}

\usepackage{times}

\usepackage{framed}

\usepackage{paralist}

\ifdefined\journal
\fi

\newcommand{\R}{\mathbb{R}}
\newcommand{\Rplus}{\R_{\geq 0}}
\newcommand{\Rplusplus}{\R_{> 0}}

\newcommand*\diff{\mathop{}\!\mathrm{d}}

\newcommand{\abs}[1]{\left\lvert #1 \right\rvert}

\newcommand{\OPT}{\ensuremath{\mathsf{OPT}}\xspace}

\renewcommand{\epsilon}{\varepsilon}

\newtheorem{theorem}{Theorem}[section]

\newtheorem{corollary}[theorem]{Corollary}
\theoremstyle{definition}

\usepackage{aliascnt}
\newaliascnt{lemma}{theorem}
\newtheorem{lemma}[lemma]{Lemma}
\aliascntresetthe{lemma}

\newaliascnt{claim}{theorem}
\newtheorem{claim}[claim]{Claim}
\aliascntresetthe{claim}

\newaliascnt{definition}{theorem}
\newtheorem{definition}[definition]{Definition}
\aliascntresetthe{definition}

\newenvironment{nestedproof}{\begin{proof}}{\end{proof}}

\ifdefined\final
\usepackage{microtype}
\newcommand{\neil}[1]{}
\newcommand{\nnote}[1]{}
\newcommand{\laura}[1]{}
\newcommand{\leon}[1]{}

\newcommand{\todo}[1]{}
\newcommand{\TODO}[1]{}

\else

\usepackage[colorinlistoftodos]{todonotes}
\let\TODO\todo

\renewcommand{\todo}[1]{{\color{red!50!black}\em \small [TODO: #1]}}

\definecolor{lightblue}{rgb}{0.38,0.82,0.90}
\colorlet{lightyellow}{yellow!60!white}
\newcommand{\neil}[1]{\TODO[color=lightyellow,inline]{\textbf{Neil:} #1}}
\newcommand{\nnote}[1]{{\color{yellow!50!black}\em \small [Neil: #1]}}
\newcommand{\laura}[1]{{\color{lightblue!70!black}\em \small [Laura: #1]}}
\newcommand{\leon}[1]{{\color{blue!50!black}\em \small [Leon: #1]}}

\fi

\newcommand{\answer}[1]{{\color{black}#1}}

\usepackage{graphicx}
\usepackage[font=small,labelfont=bf]{caption}
\usepackage{subcaption}
\usepackage{setspace}

\graphicspath{{figures/}}

\makeatletter
\renewcommand*\env@matrix[1][*\c@MaxMatrixCols c]{%
	\hskip -\arraycolsep
	\let\@ifnextchar\new@ifnextchar
	\array{#1}}
\makeatother

\newcommand{\fin}{f^+}
\newcommand{\fout}{f^-}
\newcommand{\Fin}{F^+}
\newcommand{\Fout}{F^-}

\newcommand{\Gloc}{\ensuremath{\hat G}}
\newcommand{\Omegaloc}{\ensuremath{\hat \Omega}}

\newcommand{\Eloc}{\ensuremath{\hat E}}

\newcommand{\linit}{l^\circ}
\newcommand{\hatlinit}{\answer{\hat l}}

\newcommand{\Einf}{E^{\infty}}
\newcommand{\Epot}{\tilde{E}}

\newcommand{\LFuncs}{\mathcal{L}}

\newcommand{\leps}{\l^{(\epsilon)}}

\newcommand{\Phinom}{\Phi^{\circ}}

\newcommand{\nnorm}[1]{\|#1\|}

\newcommand{\ldmax}{\kappa}
\newcommand{\Traj}{\mathcal{X}}
\newcommand{\ldmin}{\kappa_{\min}}

\newcommand{\nuinit}{\nu^{\circ}}
\newcommand{\tauinit}{{\tau^{\circ}}}
\newcommand{\Iss}{I_{ss}}
\newcommand{\Cpot}{C_{\Phi}}

\newcommand{\CCL}{\cite{cominetti2015existence}}
\newcommand{\CCO}{\cite{cominetti2021long}}

\usepackage{authblk}
\usepackage[colorlinks,linkcolor=myblue,citecolor=mygreen,hypertexnames = false]{hyperref}
\usepackage[nameinlink]{cleveref}

\crefname{claim}{Claim}{Claims}
\crefname{lemma}{Lemma}{Lemmas}
\crefname{theorem}{Theorem}{Theorems}
\renewcommand{\l}{\ell} 

\ifdefined\journal
\title{Continuity, Uniqueness and Long-Term Behavior of Nash Flows Over Time}
\author{ }
\date{ }

\else

\title{Continuity, Uniqueness and Long-Term Behavior of Nash Flows Over Time
	\thanks{N.O.\ is partially supported by NWO Vidi grant 016.Vidi.189.087.}}

\author[1]{Neil Olver}
\author[2]{Leon Sering}
\author[3]{Laura Vargas Koch}

\affil[1]{Department of Mathematics, London School of Economics and Political Science}
\affil[2]{Department of Mathematics, ETH Z\"urich}
\affil[3]{School of Business and Economics, RWTH Aachen University}
\date{ }
\fi

\begin{document}
	\ifdefined\journal
	\onehalfspacing
	\fi
	
	\maketitle
	
	\thispagestyle{empty}
	\ifdefined\journal
	\vspace{-2cm}
	\fi
	
	\begin{abstract}
		We consider a dynamic model of traffic that has received a lot of attention in the past few years.
Users control infinitesimal flow particles aiming to travel from an origin to a destination as quickly as possible.
Flow patterns vary over time, and congestion effects are modeled via queues, which form whenever the inflow into a link exceeds its capacity.
Despite lots of interest, some very basic questions remain open in this model.
We resolve a number of them \answer{in the single-commodity setting}:
\begin{itemize}
    \item We show \emph{uniqueness} of journey times in equilibria. 
    \item We show \emph{continuity} of equilibria: small perturbations to the instance or to the traffic situation at some moment cannot lead to wildly different equilibrium evolutions.
    \item We demonstrate that, assuming constant inflow into the network at the source, equilibria always settle down into a ``steady state'' in which the behavior extends forever in a linear fashion.
\end{itemize}
One of our main conceptual contributions is to show that the answer to the first two questions, on uniqueness and continuity, are intimately connected to the third.
%
To resolve the third question, we substantially extend the approach of \cite{cominetti2021long}, who show a steady-state result in the regime where the input flow rate is smaller than the network capacity. 

%
%
%

	\end{abstract}

	\pagenumbering{arabic}
	
	\section{Introduction}\label{sec:intro}

Motivated especially by congestion in transportation networks and communication networks, 
the study of routing games has received a huge amount of attention.
Most of this work concerns \emph{static} models; that is, the model posits a constant, unchanging demand, and a solution is represented by some kind of flow. 
Congestion effects are modeled via a relationship between the amount of traffic using a particular link, and the resulting delay experienced.
Many variants have been considered: nonatomic games (where each individual player controls an infinitesimal amount of flow), atomic games with a finite number of players, multi-commodity and single-commodity settings, different choices of congestion functions, and much more.
Much is understood about equilibrium behavior: conditions for existence and uniqueness; bounds on the \emph{price of anarchy} and various related notions; the phenomenon of Braess's paradox; tolling to improve equilibrium efficiency; and so on (see~\cite{RoughgardenBook} for a survey of the area).

While these static models can be a good approximation in many situations, 
this is not always the case.
More recently, there has been a lot of interest in models that are explicitly \emph{dynamic}\,---\,that is, time-varying.
A canonical situation to consider is morning rush-hour traffic; clearly there are substantial variations of traffic behavior and congestion over time. 
Another motivating setting is the routing of packets in communication networks, which traverse through a network of limited bandwidth over time, and are processed in queues at the nodes of the network. 

In the static case, models typically allow for some relationship between traffic density and delay to be posited (for example, a linear relationship, or something more refined based on empirical data).
This is substantially more difficult to do in the dynamic setting.
It is rather crucial to maintain a \emph{first-in-first-out} property for flow on an arc; overtaking is questionable from a modeling perspective, and also introduces various pathologies.
This property is unfortunately easy to violate; for example, if one attempts to specify the delay a user experiences as a function of the inflow rate into a link at the moment of entry, a sharp decrease in inflow will lead to overtaking.
Well-behaved models in this generality require detailed modeling of traffic along links (as opposed to describing the traffic on a link via a single time-varying value); an example of a model taking this approach is the LWR, or kinematic wave, model \cite{lighthill1955kinematic,richards1956shock}.
These models are very challenging to analyze, even for a single link.
The survey by Friesz and Han~\cite{friesz2019mathematical} discusses approaches towards models with quite general congestion behavior.

Another direction, with a distinct line of literature, concerns discrete packet models.
In these models, each user controls an indivisible packet that must be routed through the network; see, e.g.,~\cite{caoatomic,Ismaili2017RoutingGamesOverTimeFIFO,KulkarniM15, scarsini2018dynamic, tauer2021fifo}. 
Congestion in these models is modeled through queues.
In this paper we will be concerned with essentially the continuous analog of these packet models.
This is the \emph{fluid queueing model}, also known as the \emph{deterministic queueing model} or the \emph{Vickrey bottleneck model} \cite{vickrey1969congestion}. 

In the fluid queueing model, each link has a \emph{capacity} and a \emph{transit time}.
If the inflow rate into the link always remains below its capacity, then the time taken to traverse the link is constant, as given by the transit time.
However, if the inflow rate exceeds the link capacity for some period, a queue grows on the entrance of the link.
The delay experienced by a user is then equal to the transit time, plus whatever time is spent waiting in the queue.
As long as there is a queue present it will empty at rate given by the link capacity; depending on whether the inflow rate is smaller or larger than the capacity, this queue will decrease or increase size. 
(See \Cref{fig:link}; the notation will be fully described in \Cref{sec:model}.)
Note that the queues are considered to be \emph{vertical}, meaning they can hold an unlimited amount of flow (\emph{horizontal} or \emph{spatial} queues have been considered in~\cite{sering2019spillback,sering2020diss}). 
\begin{figure}[tbh]
	\centering
	\includegraphics{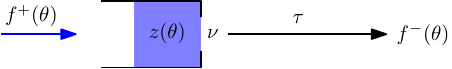}
	\caption{The dynamic of an arc at snapshot time $\theta$. The inflow $f^+(\theta)$ and outflow $f^-(\theta)$ describe the flow entering or leaving the arc at time $\theta$. The amount of flow in the queue $z_e(\theta)$ can leave the queue with rate $\nu$ and afterwards traverses the arc, which takes $\tau$ units of time.}
	\label{fig:link}
\end{figure}

Consider a network of such links, and suppose that all traffic in the network has the same origin and destination.
Clearly this is a restrictive assumption, but already this restricted setting is very challenging (moreover, we will argue in \Cref{sec:single_com} that our results extend to the more general single-commodity setting where multiple sources and sinks are allowed).
%
Starting at time $0$, users are released into the system at the origin at some constant rate;
each user controls an infinitesimal amount of flow. 
Their joint choices yield a \emph{flow over time}, as first introduced by Ford and Fulkerson in the fifties~\cite{ford_fulkerson_1956}.
But here, each user is self-interested and wants to find a \emph{quickest} route to the destination.
This will depend on the choices of the other users, since these will impact the queue lengths.
A Nash equilibrium in this model is then a joint choice of routes for all users, such that all users are satisfied in hindsight with their choices: no user can switch to an alternative route and arrive at the destination at a strictly earlier time.
Note that this means that when making a decision, it is not the queueing delays \emph{now} that matter, but the queues that the user would actually experience upon reaching an arc, which may be different.
An alternative notion of ``instantaneous dynamic equilibria''~\cite{graf2020dynamic} has also been considered, where users base decisions on current queue lengths rather than what they will actually experience, and more recent work by Graf, Harks, Kollias and Merkl~\cite{Graf_Harks_Kollias_Markl_2022} considers a much more general information framework where users use predictions of future congestion patterns.
This paper concerns only Nash equilibria.

There has been some substantial progress in understanding equilibria in this model.
Koch and Skutella \cite{koch2011nash} (see also \cite{koch2012phd}) were the first to study it using tools from combinatorial optimization.
They uncovered an intricate structure in the time derivatives of equilibria, which they could precisely describe as a solution to a certain nonlinear system (which they called \emph{thin flow with resetting}, and which we will encounter later).
Assuming solutions to this system always exist, they showed how an equilibrium could be constructed by essentially integrating; 
the resulting equilibrium has the property that all queue lengths are piecewise linear.
They call a maximal interval of departure times on which queues are affine a \emph{phase}.
(Note that there is no guarantee that there are a finite number of phases, even within a compact interval; 
that is, the algorithm is not known to be finite).
A similar integration approach can be found in the transportation literature; see Kuwahara and Akamatsu~\cite{kuwahara1993dynamic} and the survey by Iryo~\cite{Iryo13}.

Existence of equilibria is a crucial question. Zhu and Marcotte~\cite{ZM00} showed existence of equilibria for a somewhat general model where congestion satisfies a ``strong FIFO'' property; however, the fluid queing model does not have this property.
Mounce~\cite{M06} showed existence in the fluid queueing model if each possible route a user may take contains at most one bottleneck arc of finite capacity. A general existence result, without any such restrictions, was shown by Cominetti, Correa and Larr\'e~\CCL{}. 
Results for models with departure time choice and more general congestion behavior were obtained by Meunier and Wagner~\cite{MW10} and 
    Han, Friesz and Yao~\cite{HFY13}.
Here we focus our discussion on the approach taken by \CCL, as it will be most relevant to our approach.
They showed that the nonlinear system of Koch and Skutella always has a solution (in fact a unique solution), and hence that the integration procedure always succeeds in constructing an equilibrium.
Surprisingly, it is an open question as to whether the thin flow equations can be efficiently solved\footnote{The only known positive result is an algorithm on series-parallel graphs \cite{kaiser2020computation}.}; in other words, whether there is a polynomial-time algorithm to compute the next phase, given that the behavior for all previous phases has been computed. 

A good upper bound on the price of anarchy (suitably defined\,---\,it is necessary to consider average arrival time rather than average journey time as the cost function to have any hope of positive results) also remains an open question.
A constant $\frac{e}{e-1}$ bound is conjectured, and some partial progress has been made~\cite{correa2019price,bhaskar2015stackelberg}.
Cominetti, Correa and Olver~\CCO{} answered a perhaps even more basic question about equilibrium efficiency: 
if the inflow at the source is not larger than the network capacity (that is, the minimum capacity of a cut separating $s$ and $t$), 
do queue lengths and hence journey times remain bounded in an equilibrium? 
They give an answer in the affirmative; this gives at least some sense in which equilibria are well-behaved and at least not disastrously inefficient. 

In this paper, we answer some significant open problems in the model.

\paragraph{Uniqueness.}
Are equilibria unique in this model? After existence, this is among the first questions to ask about equilibria in any model.

Some care is needed in phrasing the question. 
Imagine an inflow rate of 1 at the source, and two parallel links between $s$ and $t$, both of capacity 1 or larger.
Then flow can be split arbitrarily (and in an arbitrarily time-varying way) between the two links, and the result is an equilibrium.
So there is no uniqueness at the level of flows.
Instead, the right question concerns the \emph{journey times} of users in equilibrium.
All flow particles leaving at time $\theta$ incur the same (smallest possible) journey time in an equilibrium. 
The correct uniqueness question is not about the routes chosen in equilibrium, but the costs (journey times) experienced. 
Is this unique?

This question has been discussed extensively in the literature.
    While we are focused here on the setting of a single origin and destination, Iryo~\cite{iryo2011multiple} showed that it can fail to hold if there are multiple origins and destination (though Iryo and Smith give positive results for ``unidirectional'' networks~\cite{iryo2018uniqueness}).
    Uniqueness has been demonstrated in the single-origin single-destination setting under additional assumptions of various strengths, e.g.,~\cite{kuwahara1993dynamic, MounceSmith, cominetti2015existence}.
But to the best of our knowledge, none of these works fully address the question for the bottleneck model without any additional assumptions.
A somewhat subtle\,---\,but as we will see, very important\,---\,aspect of this is explained clearly in \CCL.
They give a proof of uniqueness\,---\,\emph{but only if one restricts to right-differentiable solutions}.
All other works based on extending an equilibrium over time, to the best of our knowledge, explicitly or implicitly require this right-differentiability assumption.
First impressions might be that such a restriction is of a technical nature, and that this \emph{essentially} shows uniqueness, but this is misleading.
The precise reason for this requires more technical preliminaries, and we postpone this until \Cref{sec:overview}.
For now, we remark that there is no a priori reason why equilibria not respecting this right-differentiability requirement could not exist, nor anything ``unphysical'' about them.

\answer{We prove that equilibria are indeed unique in the single-origin single-destination setting with constant inflow, without such assumptions.} 
To further motivate why this uniqueness result is important, we will shortly discuss how it relates to the next question we resolve.

\paragraph{Continuity.}
If an equilibrium is disturbed in some way, does the disturbed equilibrium remain close in some sense to the undisturbed equilibrium, or could it veer off wildly in another direction?
Such a property is \emph{crucial} for the model to have any bearing on reality. 
Clearly, real traffic situations will not \emph{exactly} match up with equilibria in the model, even under the most optimistic assumptions.
Individual cars or packets are not really infinitesimally small; this is an approximation.
Given two routes that have slightly different but almost-equal journey times, a user may not notice or be sufficiently concerned, and pick the slightly longer one.
Actual travel delays on a link, or link capacities, could vary slightly over time due to all sorts of factors.
All of this is of course obvious\,---\,there is no expectation that the abstract, simplified model would capture all these real-world aspects.
Nonetheless, the hope is that the model is a good one, in the sense that it captures qualitatively important aspects of the real situation, \emph{and} that the closer some more complex (artificial or real-world) situation matches the conditions of the model, the closer the behavior of the model would match the more complex setting.
If slightly perturbing the situation at some moment can lead to completely different equilibrium behavior, then any conclusions drawn from the model must be treated with extreme skepticism. 
Does the model really have anything to tell us, in this case?

It might seem that continuity should be a straightforward property to show; or at least, that it should be a matter of proficiency with analytic techniques. 
Similarly to the uniqueness question, this turns out to not be the case at all, and for very similar reasons.
Again, we hold the explanation of this to the technical overview.

We demonstrate continuity of equilibria, in various senses.
We show that the equilibrium depends continuously on the initial conditions (in terms of initial queue lengths).
Alternatively, if the equilibrium is ``bumped'' at some moment in time, the equilibrium does not change much.
We also demonstrate continuity with respect to changes in the instance; if we change the capacity and/or transit times of some arcs in the network by very small amounts, the resulting equilibrium will not change much.

We should say that our results do not suffice to show the strongest forms of stability one might like, and that were hinted at in the motivation above. 
Our results only apply to settings where except for a single exceptional moment, or some finite number of moments, behavior is precisely as in a dynamic equilibrium.
They do not imply that a version of the model with small but indivisible agents behaves similarly to the nonatomic model, or that an approximate equilibrium where agents take only approximate shortest paths is similar to an exact equilibrium.
We discuss in detail such questions, as well as followup work addressing them, in the conclusion (\Cref{sec:conclusion}).



\medskip

It should seem natural, given the above discussion, that uniqueness and continuity are related.
The following question, however, appears entirely unrelated. 
Surprisingly, this is not the case: one of our main conceptual contributions is to show how uniqueness and continuity follow from resolving it.

\paragraph{Long-term behavior.}
As already mentioned, \CCO{} prove that if inflow at the source is constant (starting from some initial time) and not larger than the capacity of the network, as measured by the minimum capacity of an $s$-$t$-cut, then queues remain bounded in an equilibrium.
Their proof actually shows something stronger:
Namely, they show that as long as this inflow condition is satisfied, after some (instance dependent) time the equilibrium will reach a ``steady state'', in which all queues remain constant from this time forward.

This raises a question. What happens if the inflow rate is larger than the minimum $s$-$t$-cut capacity? 
Clearly, queues can no longer stay bounded.
But one may still ask whether the evolution eventually reaches a steady state\,---\,which no longer means a situation where no queues change (clearly impossible), but a situation where queues change linearly, forever into the future.

Just like the \CCO{} result, this can be viewed as a positive statement about the efficiency of equilibria.
It turns out that there is only one possible rate vector at which queues can grow in a steady state.
So if steady state is reached, it means that over a sufficiently long time horizon, the equilibrium does use the network in the most efficient way possible; queues only grow because they must, due to bottlenecks in the network.
Further, this queue behavior at steady state can be efficiently characterized.

While this result is already interesting in its own right, a very slight generalization of this also turns out to be the main technical ingredient in our proof of uniqueness and continuity.
Our connection between these two questions shows in fact that continuity is \emph{equivalent} to convergence to steady-state in an appropriately defined subnetwork of the instance.
Our proof of convergence to steady state is, like the one of \CCO{}, based on a non-obvious potential related to a primal-dual program that characterizes possible steady-state situations.
In this sense our result on continuity, despite initial appearances, is fundamentally combinatorial rather than analytic.

\paragraph{Piecewise-constant inflows.}
While we have described the model and our results in the context of constant inflow, it is rather straightforward to see that continuity and uniqueness of equilibria with piecewise-constant inflows is an immediate consequence. 

Of course, our result on reaching steady state does not hold for piecewise-constant inflows in general\,---\,if the inflow is varying forever, we cannot expect queues to behave linearly.

\paragraph{Single-commodity instances.}
All our results can be extended to the more general single-commodity setting. 
In particular, they apply if all traffic has a common destination, but comes from multiple sources, each with a given inflow rate.
This is of course very far from the generality where users have distinct origins and destinations, but is quite relevant for modeling morning peak traffic entering the central business district of a city.

	\section{Model and preliminaries} \label{sec:model}

An instance is described by a directed graph $G = (V, E)$ with a source $s$ and sink $t$, where each arc is equipped with a transit time $\tau_e \geq 0$ and a capacity $\nu_e > 0$. At $s$ we have a constant \emph{network inflow rate} of $u_0 > 0$, which begins at time $0$.
We may assume that every node in $G$ is both reachable from $s$, and can reach $t$.
For technical convenience, we will follow previous works and assume that $G$ has no directed cycle consisting of arcs with transit time $0$.

We use the notation $\delta^-(v)$ and $\delta^+(v)$ to denote the set of incoming and outgoing arcs at $v$, respectively, and similarly $\delta^-(S)$ and $\delta^+(S)$ for arcs entering or leaving a set $S$. Moreover, let $E[S]$ denote arcs with both end points in a set $S$.
We define $[z]^+ \coloneqq \max\{0, z\}$.

\subsection{Flows over time}

A \emph{flow over time} is given by a family of locally-integrable 
functions $\fin_e : \Rplus \to \Rplus$ and $\fout_e: \Rplus \to \Rplus$ that 
describe the inflow and outflow rate at each arc $e \in E$ at every point in time.
The \emph{cumulative} inflows and outflows are given by the absolutely continuous functions $\Fin_e(\xi) \coloneqq \int_0^\xi \fin_e(\xi')\diff\xi'$ and
$\Fout_e(\xi) \coloneqq \int_0^\xi \fout_e(\xi')\diff\xi'$, respectively.
A flow particle in the queue of link $e$ at time $\xi$ must have entered the link at time $\xi$ or before, and will not leave the link until after time $\xi+\tau_e$.
Thus the \emph{queue volume} on a link $e$ at time $\xi$, denoted by $z_e(\xi)$, is 
given by
$z_e(\xi) \coloneqq \Fin_e(\xi) - \Fout_e(\xi+\tau_e)$.
Since queues always empty at maximum rate, the amount of time spent waiting in a queue is simply the queue volume upon arrival divided by the arc capacity.
%
We call a flow over time \emph{feasible} if it satisfies flow conservation at every node $v \in V$ for every point in time $\xi$:
\[\sum_{e \in \delta^+(v)} \fin_e(\xi) - \sum_{e \in \delta^-(v)} \fout_e(\xi) \begin{cases} = 0 & \text{ for } v \in V \setminus \set{s, t},\\
	= u_0 & \text{ for } v = s,\\
	\leq 0 & \text{ for } v = t,
\end{cases}\]
and if the queues empty at a rate given by the capacity:
\[
\fout_e(\xi+\tau_e) = \begin{cases} \nu_e & \text{ if } z_e(\xi) > 0,\\
	\min\set{\nu_e, \fin_e(\xi)} & \text{ otherwise.} 
\end{cases}
\]
We will start at time $0$ with an empty network, requiring that in a feasible flow $f_e^-(\theta)=0$ for all $\theta<\tau_e$ and all $e \in E$.
The \emph{earliest arrival time} at $w$ of a particle starting at $s$ at time $\theta \in \Rplus$ is given by
\[\l_w(\theta) \coloneqq \begin{cases}
	\theta & \text{if } w=s,\\
	\min\limits_{e = vw \in \delta^-(w)} \l_v(\theta) + \frac{z_e(\l_v(\theta))}{\nu_e} + \tau_e& \text{otherwise.}
\end{cases}
\]
We will often refer to the collection of earliest arrival times for a fixed $\theta$ as a \emph{labeling}. 
Given earliest arrival times $\l$, we define the associated \emph{queueing delays} (or \emph{queueing delay functions}) by $q_e(\theta) \coloneqq z_e(\l_v(\theta)) / \nu_e$ for each $e=vw \in E$.
Note that $q_e(\theta)$ is the queueing delay on arc $e$ for a flow particle departing $s$ at time $\theta$ and taking a shortest path to $e$.
For any labeling $\l(\theta)$, we define the \emph{active} arcs $E'_{\l(\theta)}$ and the \emph{resetting} arcs $E^*_{\l(\theta)}$ as 
\[ E'_{\l(\theta)} \coloneqq \set{e = vw \in E | \l_w(\theta) = \l_v(\theta) + q_e(\theta) + \tau_e}
\quad \text{and} \quad 
E^*_{\l(\theta)} \coloneqq \set{e = vw \in E | q_e(\theta) > 0} \]
respectively.
We will write $E'_\theta$ and $E^*_\theta$ as shorthand, if the choice of $\l$ is clear.
So $E'_\theta$ consists of all arcs that lie on a shortest path from $s$ to some node in the network, from the perspective of a user that departs at time~$\theta$; and $E^*_\theta$ is the set of arcs where such a user would find a queue if they enter the arc as early as possible.

\subsection{Equilibria}

An \emph{equilibrium} (also referred to as a \emph{dynamic equilibrium} or a \emph{Nash flow over time}) is a feasible flow over time in which almost all flow particles travel along a shortest path from $s$ to $t$, i.e., only along active arcs.
In this case, resetting arcs are always active and the active and resetting arcs are characterized as follows (see \cite[Proposition 2]{cominetti2015existence}):
\begin{equation}\label{eq:actres}
	E'_{\l(\theta)} = \set{e = vw \in E | \l_w(\theta) \geq \l_v(\theta) + \tau_e} \hspace{0.5em}
	\text{and} 
	\hspace{0.5em}
	\, E^*_{\l(\theta)} = \set{e=vw \in E | \l_w(\theta) > \l_v(\theta) + \tau_e}.
\end{equation}
For an arc $e=vw \in E^*_\theta$, the delay experienced by such a user departing at time $\theta$ is given by $q_e(\theta) = \l_w(\theta) - \l_v(\theta) - \tau_e$.

As proven in \CCL, a feasible flow over time is an equilibrium if and only if $\Fin_e(\l_v(\theta)) = \Fout_e(\l_w(\theta))$ for all arcs $e = vw$ and all departure times $\theta$. 
Define $x_e(\theta) \coloneqq \Fin_e(\l_v(\theta))$ for all $e$ and $\theta$. 
For an equilibrium flow, the derivative of $x_e$ at time $\theta$, which exists almost everywhere, can be interpreted as a flow describing what proportion of flow particles departing at time $\theta$ use arc $e$; $x'$ is an $s$-$t$-flow of value $u_0$ almost everywhere.
It has been shown~\cite{koch2011nash} that for an equilibrium, the resulting pair $(\l ,x)$ satisfies the following for almost every $\theta$,
setting $x' = x'(\theta)$, $\l' = \l'(\theta)$, $E' = E'_\theta$ and $E^* = E^*_\theta$.

\begin{subequations}
	\makeatletter\def\@currentlabel{TF}
	\makeatother
	\label{eq:tf-group}
	\renewcommand{\theequation}{TF-\arabic{equation}}
	\begin{alignat}{2}
		x' &\text{ is a static $s$-$t$-flow of value $u_0$ on $E'$}, \label{eq:tf-flow}\\
		\l'_{s} &= 1, && \label{eq:l'_s:base}\\ 
                \l'_w &\leq \rho_e(\l'_v, x'_e) \quad &&  \text{ for all } e=vw \in E',
		\label{eq:l'_v_min:base}\\ 
            \l'_w &= \rho_e(\l'_v, x'_e) && \text{ for all } e = vw \in E' \text{ with } x'_e > 0, 
                \label{eq:l'_v_tight:base}
	\end{alignat}
\end{subequations}
\addtocounter{equation}{-1} 
\[\text{ where } \qquad \rho_e(\l'_v, x'_e) \coloneqq \begin{cases}
	\frac{x'_e}{\nu_e} & \text{ if } e = vw \in E^*,\\
	\max\bigl\{\l'_v, \frac{x'_e}{\nu_e}\bigr\} & \text{ if } e = vw \in E'\backslash E^*.  
\end{cases}\]
These are called the \emph{thin flow} conditions (sometimes \emph{thin flows with resetting}) for the configuration $(E', E^*)$.
\cite{cominetti2011existence} replaced \eqref{eq:l'_v_min:base} with the stronger requirement
\begin{equation}\label{eq:ntf}\tag{TF-3$'$}
    \l'_w = \min_{e=vw \in E'} \rho_e(\l'_v, x'_e) \qquad \text{ for all } w \in V \setminus \{s\}, 
\end{equation}
and showed that the pair $(\l,x)$ corresponding to a flow over time is an equilibrium if and only if these strengthened \emph{normalized thin flow conditions} for configuration $(E'_\theta, E^*_\theta)$ are satisfied by $(\l'(\theta), x'(\theta))$ for almost every $\theta$. 
In particular, $\l'(\theta)$ and $x'(\theta)$ exist for almost every $\theta$; we call such values of $\theta$ \emph{points of differentiability}. 


We will take the viewpoint throughout that an equilibrium \emph{is} a pair $(\l, x)$ whose derivative satisfies the normalized thin flow equations almost everywhere. 
Such a pair fully determines the flow over time $(\fin, \fout)$ of the equilibrium; 
for example, the queue volume is determined by $\tfrac1{\nu_e}z_e(\ell_v(\theta)) = \max\{0, \ell_w(\theta) - \ell_v(\theta) - \tau_e\}$ for all $e=vw$ and $\theta \in \Rplus$.
Note that from this perspective, given $(\l, x)$ satisfying the normalized thin flow conditions for all times $\theta \in [0, T]$, extending the equilibrium to later departure times only requires knowing $\l(T)$; nothing about the earlier (with respect to departure time) behavior of the equilibrium is needed.
Just $\l$ alone captures all the truly important information; one could easily verify whether a given curve $\l$ can be extended to an equilibrium, since given $\l'(\theta)$, it is easy to determine whether a matching $x'(\theta)$ exists so that $(\l'(\theta), x'(\theta))$ satisfy the normalized thin flow conditions.


The thin flow conditions, and the normalized thin flow conditions, are fully determined by the configuration $(E', E^*)$.
However not all choices of $(E', E^*)$ are meaningful.
\begin{definition}
We call a configuration $(E', E^*)$ with $E^* \subseteq E' \subseteq E$ a \emph{valid configuration} if 
	\begin{compactenum}[(i)] 
		\item for every node $v \in V$ there is an $s$-$v$-path in $E'$, \label{valid:i}
		\item every arc $e \in E^*$ lies on an $s$-$t$-path in $E'$, and \label{valid:ii}
		\item no arc of $E^*$ lies on a directed cycle in $(V,E')$. \label{valid:iii}
\end{compactenum}
A vector $l \in \R^V$ is called \emph{feasible} if $(E'_l, E^*_l)$ is a valid configuration.
\end{definition}
If $(\l,x)$ is an equilibrium, then $\l(\theta)$ is a feasible label for every $\theta$.
Condition (\ref{valid:i}) is simply insisting that $\l(\theta)$ does correspond to earliest arrival times from $s$, for some choice of nonnegative queue lengths on the arcs. 
Condition (\ref{valid:ii}) follows because by the definition of dynamic equilibrium, a flow particle at the back of a queue travels on an active $s$-$t$-path. 
The sum of transit times along a directed cycle is assumed to be strictly positive, meaning that $E'_{\l(\theta)}$ is acyclic and implying (\ref{valid:iii}).

\CCL{} showed that for any valid configuration $(E', E^*)$ with $E'$ acyclic, the normalized thin flow conditions have a solution;
    from this they could deduce existence of equilibria.
    They also showed that the solution is \emph{label-unique}: any other solution must have the same values for $\l'_v$ for all $v$.
    It follows that in an equilibrium, $\l$ is a piecewise linear function; each linear segment is referred to as a \emph{phase} of the equilibrium~\cite{koch2011nash,cominetti2015existence}.
    \nnote{Removed finite number of phases problem mention, to avoid messing with the flow; we discuss anyway in the conclusion.}

    We will require a slightly stronger result however: we will need existence and label-uniqueness for \emph{any} valid configuration, without the restriction that $E'$ be acyclic. 
    \answer{(This is not strictly needed for our uniqueness and continuity result, but it will be important for our result on long-term behavior for general networks, as we will later see.)}
    Existence was shown by Koch~\cite[Theorem 6.61]{koch2012phd}. 
    Label-uniqueness fails to hold, however, even with the inclusion of the normalization condition \eqref{eq:ntf}.%
\footnote{
As an explicit example, consider a network $G$ with arcs $\{ sv, vt, vw, wy, yz, zw, zs\}$ and capacities $\answer{\nu}_{sv} = \tfrac12$, $\answer{\nu}_e = 1$ for all other arcs $e$, with an inflow of $1$. It is easy to see that assigning $\l'_s = 1$, $\l'_v = \l'_t = 2$ and $\l'_w = \l'_y=\l'_z = \alpha$ is a solution to the normalized thin flow equations for configuration $(E,\emptyset)$ for all $\alpha \in [1,2]$.}

    Instead, we will use a slightly different strengthening of the thin flow conditions, and one that is completely equivalent if $E'$ is acyclic.
We will always consider a \emph{maximal} solution to the thin flow equations \eqref{eq:tf-group}, by which we mean a solution $(\l', x')$ for which $\l'$ is pointwise maximal.
Observe that this implies \eqref{eq:ntf}: if $(\l', x')$ is any solution to \eqref{eq:tf-group}, and $w$ a node for which
$\l'_w < \rho_e(\l'_v, x'_e) $ for all $e=vw \in E'$, then none of these arcs is flow carrying (by \eqref{eq:l'_v_tight:base}), and so increasing $\l'_w$ by a sufficiently small amount yields a new valid thin flow solution, showing that $(\l', x')$ was not maximal.
If $E'$ is acyclic, the two notions are in fact identical, given that the normalized thin flow equations have a label-unique solution.
But we now have label-uniqueness of maximal thin flow solutions for \emph{any} valid configuration.
The proof follows exactly the same lines as the proof in \CCL{} for the case where $E'$ is acyclic, with minor modifications.



%
\begin{lemma}
 For any valid configuration $(E',E^*)$, maximal solutions to \eqref{eq:tf-group} are label-unique. 
\end{lemma}

\begin{proof}
Suppose for the sake of contradiction there are two maximal solutions to \eqref{eq:tf-group}, namely $(\l',x')$ and $(h',y')$ with $\l'\neq h'$.
Let $S\coloneqq \{ v \in V \mid \l_v' > h'_v\}$; this set is non-empty as otherwise $\l'$ would not be pointwise maximal. 

\begin{claim}
	\label{claim:no_flow_across_S}
    $x_e'=y_e'=0$ for all $e \in \delta^+(S) \cup \delta^-(S)$. 
\end{claim}

\begin{nestedproof}
 Suppose first that $x_e'>y_e'\geq 0$ for some $e \in \delta^+(S)$. Then $\l_w'=\rho_e(\l_v',x_e')>\rho_e(h_v',y_e')\geq h_w'$, a contradiction as $w\not\in S$.
 Next, suppose 
 $x_e' < y_e'$ and for some $e \in \delta^-(S)$. Then $h_w'=\rho_e(h_v',y_e')\geq\rho_e(\l_v',x_e')\geq l_w'$ (using $y'_e > 0$ in the first equality), again a contradiction as $w\in S$.
 
 The above yields that $x_e'=y_e'$ for all $e \in \delta^+(S) \cup \delta^-(S)$ as both $x'$ and $y'$ are static flows of the same flow value. Now, observe that $y_e'>0$ for $e \in \delta^-(S)$ would still yield the contradiction $h'_w \geq \l'_w$. 
 Thus $x'_e = y'_e = 0$ for all $e \in \delta^-(S)$, and as $s \not\in S$ this yields $x'_e = y'_e = 0$  for all $e \in \delta^+(S)$ as well.
\end{nestedproof}
 
 
 Let 
 \[
 k'_v = \begin{cases}
 	\l_v' \quad &\text{if } v \in S\\
 	h_v' \quad &\text{otherwise,}
 \end{cases}
 \qquad  \qquad 
  \bar x'_e = \begin{cases}
 	x_e' \quad &\text{if } e \in E[S]\\
 	y_e' \quad &\text{otherwise.}
 \end{cases}
 \]
We show that $(k',\bar x')$ is a thin flow for configuration $(E', E^*)$. 
By \Cref{claim:no_flow_across_S}, $\bar x'$ is an $s$-$t$-flow of value $u_0$ on $E'$. 
\eqref{eq:l'_s:base} and \eqref{eq:l'_v_tight:base} clearly remain valid, so we just have to make sure that \eqref{eq:l'_v_min:base} is not violated by any edges crossing $S$. 
Thus, consider some entering edge $e=vw \in \delta^-(S)$. \eqref{eq:l'_v_min:base} requires that $k_w' \leq \rho_e(k_v',0)=k_v'$.  This follows as $k_w' =\l_w' \leq \rho(\l_v',0) =\l_v'\leq h_v'=k_v'$.
Similarly, if we consider some leaving edge $e=vw \in \delta^+(S)$, \eqref{eq:l'_v_min:base} requires that $k_w' \leq \rho_e(k_v',0)=k_v'$.  This follows as $k_w' =h_w' \leq \rho\answer{_e}(h_v',0) =h_v' < l_v'=k_v'$. Thus, $(k', \bar x')$ is indeed a thin flow.
Since $k' \geq \l'$ and $k' \geq h'$, we obtain a contradiction to the assumption that $(\l',x')$ and $(h',y')$ are distinct maximal solutions.
\end{proof}
    
From now on, we only consider maximal solutions to the thin flow equations. If $(\l', x')$ is a maximal solution, we will refer to $\l'$, which is unique, as the \emph{thin flow direction}.


\answer{
The following simple fact will be useful.
\begin{lemma}\label{lem:labelsinc}
	For any valid configuration $(E', E^*)$ and corresponding maximal thin flow solution $(\l', x')$, $\l'_v > 0$ for all $v \in V$. 
       Hence  $x'_e > 0$ for every $e \in E^*$.
\end{lemma}
\begin{proof}
        Define $V_0 \coloneqq \set{ v \in V | \l'_v = 0}$.
        Note that $\sum_{e \in \delta^-(v)} x'_e = 0$ for all $v \in V_0$ by \eqref{eq:l'_v_tight:base}, which immediately implies  $t \notin V_0$ as $\sum_{e \in \delta^-(t)} x'_e = u_0>0$. 
        Moreover, as $\l_s'=1$ and thus $s\not\in V_0$, flow conservation yields that $x_e'=0$ for all $e \in \delta^+(V_0)$.
	
	Assume for contradiction that $V_0 \neq \emptyset$.
		As we choose a maximal thin flow solution, there has to be a resetting arc $f = yz$ with $x'_{f} = 0$.
            Otherwise, $\l' $ is not pointwise maximal: with $\epsilon$ equal to the smallest nonzero label of $\l'$, $\bar\l'$ defined by $\bar\l'_v = \max\{ \l'_v, \epsilon\}$ for all $v$ is easily seen to still be a solution to the thin flow equations.
		
	Note that $z\in V_0$. 
        As $(E', E^*)$ is a valid configuration, there has to be an active $z$-$t$-path. 
        In particular there has to be an active arc $e = vw \in E'$ with $v \in V_0$ and $w \notin V_0$. But since $x'_e = 0$, 
	this is a contradiction to \eqref{eq:l'_v_min:base}.

        The second part of the claim is an immediate consequence of \eqref{eq:tf-group}; $x'_e = 0$ for $e=vw \in E^*$ would imply $\l'_w = 0$.
\end{proof}
}


\subsection{Equilibrium trajectories and generalized networks}

\answer{
    The exposition so far demonstrates the primacy of the earliest arrival labels. 
    If $X(l)$ is the thin flow direction for configuration $(E'_l, E^*_l)$ wherever this is valid, 
    then an equilibrium $(x,\ell)$ satisfies that $\l'(\theta) = X(\l(\theta))$ for almost all $\theta \in \Rplus$.
    This is the primary perspective on equilibria that we will take going forward. 
    However, we will need to generalize it in two ways, for the purposes of our later arguments.

    Firstly, we will allow for an arbitrary starting point, not necessarily the labeling corresponding to the empty network.
    Secondly, we allow for a specified set of ``free arcs'' $\Einf \subseteq E$ to be given which  
            will \emph{always} be resetting (and hence also always active). In some sense, these represent arcs where the queue is allowed to be negative. 
            The reason that this will be useful will become clearer later; for now, we focus on the precise definitions.

    \begin{definition}
        A \emph{generalized network} is a network $G=(V,E)$, along with a set $\Einf \subseteq E$ of \emph{free arcs} such that $(E, \Einf)$ is a valid configuration.
    \end{definition}
    In particular, $\Einf = \emptyset$ is always a valid choice. 
    Fix a generalized network $G=(V,E)$ and $\Einf$ for the remainder of this section.

    We redefine $E'_l$ and $E^*_l$ for $l \in \R^V$ by
    \begin{equation}\label{eq:gen-actres}
    E'_l = \Einf \cup \{ e=vw \in E \mid l_w - l_v \geq \tau_e\} \qquad\text{and}\qquad E^*_l = \Einf \cup \{ e=vw \in E \mid l_w - l_v > \tau_e \}. 
\end{equation}
    As before, $l$ is a feasible labeling if $(E'_l, E^*_l)$ is a valid configuration.
    Let $\Omega$ denote the set of feasible labelings.

    \todo{I moved this here, since we need it also in Lemma 2.8. It seems not worth writing a full lemma for, is it ok like this?}\laura{perfect}
    We remark that $E'_l$ is acyclic for any $l \in \Omega$, just as was the case for normal networks. 
    This follows as the set $\{e=vw \in E \mid l_w-l_v-\tau_e \geq 0\}$ is acyclic by our assumption that sum of transit times along a directed cycle is positive, 
    and moreover, no arc $\Einf$ lies on a cycle of active arcs as $(E_l', E_l^*)$ is a valid configuration.

}


\answer{
    The following lemma shows that a trajectory that follows the vector field $X$ will not leave $\Omega$.
(The existence results in the literature certainly imply this in the case that $\Einf = \emptyset$ and $\l(0)$ is the labeling of an empty network.) 
\begin{lemma}\label{lem:X_points_inside}
    For every $l \in \Omega$, $X(l) \in \{ d \in \R^V : l + \epsilon d \in \Omega \text{ for some } \epsilon > 0\}$.
\end{lemma}
\begin{proof}
    Fix $l \in \Omega$, and let $l' = X(l)$.
    Let $\epsilon > 0$ be small enough that with $\bar{l} = l + \epsilon l'$, for any $e=vw$ with $l_w - l_v - \tau_e \neq 0$, $\bar{l}_w - \bar{l}_v - \tau_e \neq 0$ and has the same sign.
    We need to show that properties \eqref{valid:i}--\eqref{valid:iii} for a valid configuration hold for $(E'_{\bar{l}}, E^*_{\bar{l}})$.
    Note that $E^*_{\bar{l}} \supseteq E^*_l$ and $E'_{\bar{l}} \subseteq E'_l$. 

    \begin{enumerate}[(i)]
        \item Suppose for a contradiction that there is no $s$-$v$ path for some $v \in V$, chosen minimally in the topological order described by $E'_{\bar{l}}$. 
            Then no arc entering $v$ is in $E'_{\bar{l}}$, and hence in $E^*_{l}$.
            By \eqref{eq:l'_v_min:base} (which holds since $E'_{l}$ is acyclic), and as no arcs entering $v$ are in $E^*_l$, there exists some $uv \in E'_l$ with $l'_v = l'_u$.
            But then $\bar{l}_v - \bar{l}_u - \tau_e = l_v - l_u - \tau_e \geq 0$, 
            contradicting out assumption that $uv \notin E'_{\bar{l}}$.

        \item Let $e=vw \in E^*_{\bar{l}}$.
            Take $(l', x')$ to be any complete solution to the thin flow equations.
            Either $l'_w > l'_v$, or $e \in E^*_l$. Either way (with the help of \Cref{lem:labelsinc} in the latter case), $x'_e > 0$. 

            Let $P$ be an $s$-$t$-path including $e$ in the support of $x'$.
            The thin flow equations imply that each $f=yz \in P$ is either in $E^*_l$ (and hence $E^*_{\bar{l}}$), or is active and has $l'_z \geq l'_y$; so this path is in $E'_{\bar{l}}$.

        \item
            Since $E'_{\bar{l}} \subseteq E'_l$ and $E'_l$ is acylic, this is immediate. 
    \end{enumerate}
\end{proof}

\answer{\begin{definition}\label{def:equilibrium_trajectory}
        An \emph{equilibrium trajectory} $\l$ is a vector $(\l_v)_{v \in V}$, with $\l_v: \Rplus \to \R^V$ differentiable almost everywhere, such that $\l(\theta) \in \Omega$ and $\l'(\theta) = X(\l(\theta))$ for all $\theta \in \Rplus$.
        The \emph{starting point} of $\l$ is the labeling $\l(0)$.
%
\end{definition}
The previous lemma implies that for any $\linit \in \Omega$, an equilibrium trajectory starting from $\linit$ can be constructed, by solving the differential equation $\l(0) = \linit$, $\l'(\theta) = X(\l(\theta))$ for $\theta \in \Rplus$, exactly as was done for equilibria where $\linit$ is the labeling for the empty network.}
}

Notice that since $X(\linit)$ depends only on $E'_{l}$ and $E^*_{l}$, we can deduce that $X$ is piecewise constant.
The regions on which $X$ is constant also have a very simple structure, given by \eqref{eq:gen-actres}: 
each arc $e=vw$ divides $\Omega$ into two open halfspaces separated by the hyperplane $\{ l \in \Omega: l_w - l_v = \tau_e \}$; a region is determined by the choice of sign (positive, negative or zero) for each link (not all combinations necessarily yield a region); see \Cref{fig:vector_field}.

\answer{
Intuitively, an equilibrium trajectory starting from $\linit$ corresponds to an equilibrium, but with some initial conditions defined by $\linit$. 
However, interpreting this in terms of  artificial initial queues involves some subtleties. 
We emphasize that no such interpretation is needed; equilibrium trajectories with initial conditions not corresponding to an empty network, or with free arcs, are tools for our analysis.
}

\begin{lemma}\label{lem:lipschitz}
    Any equilibrium trajectory is $\ldmax$-Lipschitz, where
    \[ \ldmax \coloneqq \max \Bigl\{ 1, u_0 / \min_{e \in E} \nu_e\Bigr\}. \]
\end{lemma}
\begin{proof}
    \answer{
    Consider any equilibrium trajectory $\l$.
    It suffices to show that for any point of differentiability $\theta$, $\l'_v(\theta) \leq \ldmax$ for all $v \in V$.
    Let $(\l', x')$ be a complete maximal thin flow solution for the configuration $(E'_{\l(\theta)}, E^*_{\l(\theta)})$ (so $\l' = \l'(\theta)$).
    Then for any $e=vw \in E'_{\l(\theta)}$, 
    \[ \l'_w \leq \rho_e(\l'_v, x'_e) \leq \max\{\l'_v, x'_e / \nu_e \} \leq \max\{ \l'_v, u_0 / \min_{e' \in E} \nu_{e'}\}, \]
    where the last inequality uses that $x'_e \leq u_0$ by the acyclicity of $E'_{\l(\theta)}$.
    Considering any active $s$-$v$-path immediately demonstrates the claim.
    }
\end{proof}

	\section{Technical overview}\label{sec:overview}

\begin{figure}[t]
	\begin{minipage}[t]{.32\textwidth}
		\centering
		\includegraphics{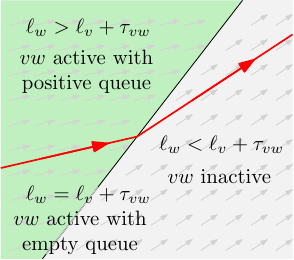}
		\caption{Dynamic equilibria can be seen as trajectories in $\R^V$ that follow a piecewise-constant vector field.}
		\label{fig:vector_field}
	\end{minipage} \hfill
	\begin{minipage}[t]{.32\textwidth}
		\centering
		\includegraphics{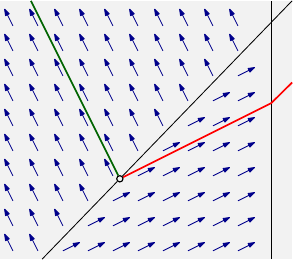}
		\caption{The simplest situation which would produce non-uniqueness (if it were possible).}
		\label{fig:diverging_vector_field}
	\end{minipage} \hfill
	\begin{minipage}[t]{0.32\textwidth}
		\centering
		\includegraphics{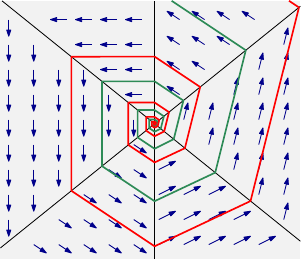}
		\caption{A more subtle potential situation which would produce non-uniqueness; we show that this cannot occur.}
		\label{fig:spiraling}
	\end{minipage}
\end{figure}

\paragraph{Uniqueness and continuity of equilibrium trajectories.}
We begin our discussion by considering uniqueness of equilibrium trajectories.

To understand why uniqueness is a strong property, and nontrivial to prove, let us consider some other piecewise constant vector fields, restricting ourselves to two dimensions.
First, consider the situation shown in \Cref{fig:diverging_vector_field}.
Non-uniqueness of trajectories is quite obvious in this case, as we have two regions with a common boundary and the vector fields in each region pointing away from this boundary. Previous work does in fact rule this out for the vector field of interest to us\,---\,it can be shown to be a consequence of the uniqueness of solutions to the thin flow conditions. 

A more difficult situation is shown in \Cref{fig:spiraling}, where the trajectory (informally) spirals outwards.
Again, non-uniqueness is rather clear; two possible trajectories starting from the origin are shown.
This possibility is \emph{not} excluded by previous work.
In particular, if a trajectory of this type existed, it would provide an example of an equilibrium $(\l, x)$ for which $\l$ is not right-differentiable.
By contrast, the procedure of Cominetti, Correa and Larr\'e~\CCL{} and Koch and Skutella~\cite{koch2011nash} produces an equilibrium whose labeling is right-differentiable, and so this would immediately be an example of non-uniqueness.

\todo{Updated to state for generalized networks, and to include the instance.}
\answer{
\begin{theorem}
	\label{thm:uniqueness}
	Given a generalized network $G=(V,E)$ with free arcs $\Einf \subseteq E$ and a feasible labeling $\linit$, there is a unique equilibrium trajectory $\ell$ starting from $\linit$.
\end{theorem}
This theorem with $\Einf=\emptyset$ and $\linit$ the labeling of the empty network gives label-uniqueness of equilibria. }

Before discussing how we prove this result, we first discuss continuity, which turns out to be intimately related.
For now, we are referring to continuity of the equilibrium trajectory, as a function of the initial feasible labeling.
This is very natural from the perspective of equilibrium trajectories, 
but care is needed in interpreting this;
for example, slightly modifying a queue or queues at some moment in the evolution of an equilibrium does not immediately correspond to a change in the labeling. 
We will discuss other more ``interpretable'' continuity results at the end of this section.

Let $\LFuncs$ denote the extended Banach space of $\ldmax$-Lipschitz functions from $[0, \infty)$ to $\R$ imbued with the {extended}\footnote{{An extended norm is allowed to take on infinite values, and otherwise behaves as a norm.}} norm $\|g\|_\infty = \sup_{\theta \geq 0} g(\theta)$.
{(Recall that $\ldmax$ is a Lipschitz constant for all equilibrium trajectories  by \Cref{lem:lipschitz}.)}
By $\Traj$ we denote the space $\LFuncs^V$ imbued with the {extended} norm $\|\l\| \coloneqq \max_{v \in V} \|\l_v\|_\infty$.
Our continuity result is the following.
\begin{theorem}\label{thm:label-continuity}
    \answer{Consider a generalized network $G=(V,E)$ with free arcs $\Einf \subseteq E$ and a feasible labeling $\linit$.}
	Let $\Psi: \Omega \to \Traj$ be the map that takes $\linit \in \Omega$ to the unique equilibrium trajectory $\l$ starting from $\linit$.
	Then $\Psi$ is a continuous map.
\end{theorem}
We remark that the continuity of $\Psi$ is a statement over the entire evolution of a trajectory\,---\,a rather strong form of continuity, and in particular, stronger than convergence of the trajectory at any fixed time $\theta$.
Our result says that if we look at the equilibrium trajectory $\l$ starting from $\linit$, and
then look at the equilibrium trajectory $\tilde{\l}$ starting from a slight perturbation of $\linit$, then
$\l$ and $\tilde{\l}$ stay close forever, rather than possibly drifting apart very slowly.

\paragraph{The connection to long-term behavior.}
To understand the connection between long-term behavior and uniqueness and continuity, we consider the ``local'' situation in a sufficiently small neighborhood of some $\linit \in \Omega$.
\answer{This is where the notion of a generalized network will become useful. 
All hyperplanes that do not go through this point are considered to be very far away, and hence are ignored. 
Call an arc $e=vw$ \emph{tight} (with respect to $\linit$) if $\linit_w - \linit_v = \tau_e$; so the hyperplane of arc $e$ passes through $\linit$ precisely if it is tight.
All non-tight arcs either have a queue or are inactive. 
Define the \emph{local network} at $\linit$ to be the generalized network $\Gloc=(V, \Eloc)$ with free arcs $\Einf$, where 
$\Eloc = E'_{\linit}$ and $\Einf = E^*_{\linit}$. 
In other words, all arcs with a queue become free arcs, and all inactive arcs are removed.
Only arcs of $\Eloc \setminus \Einf$ have an associated hyperplane; all these hyperplanes pass through $\linit$.
Thus the vector field $\hat{X}$ of the local network has a conic self-similarity: $\hat{X}(\linit + \alpha d) = \hat{X}(\linit + d)$ for any $d \in \R^V$ with $\linit+d \in \Omegaloc$ and $\alpha > 0$,
where $\Omegaloc$ denotes the set of feasible labelings for the local network.}
\todo{I changed $X$ to $\hat{X}$, seemed more consistent. I don't think it's used anywhere else?}

We are now ready to make the connection between uniqueness/continuity and long-term behavior of equilibrium trajectories in this conic setting.
Let $(y, \lambda)$ be the solution of the thin flow problem at $\linit$; so $\lambda = X(\linit)$.
Suppose that we could show that for some given initial condition $\hatlinit$, any equilibrium trajectory $\l$ starting from $\hatlinit$ eventually, after some finite amount of time $T$, satisfies $\l'(\theta) = \lambda$ for all $\theta \geq T$. 
Here, $T$ may certainly depend on the choice of $\hatlinit$.
But now we can exploit the conic symmetry.
Consider an equilibrium trajectory $\leps$ starting from $\linit + \epsilon(\hatlinit - \linit)$.
Then it is quite easy to see that the trajectory $\l$ defined by $\l(\theta) \coloneqq (\leps(\epsilon \theta) - \linit) / \epsilon + \linit$ is an equilibrium trajectory starting from $\hatlinit$.
Thus $\leps(\theta)$ will reach steady state at time $\epsilon T$.
So as we move the initial condition closer and closer to $\linit$, any equilibrium trajectory looks more and more like the trajectory 
$\l^*$ given by $\l^*(\theta) = \linit + \theta \lambda$.
The maximum distance between $\leps(\epsilon)$ and $\l^*(\epsilon)$ can be controlled by the distance between the initial condition and $\linit$ and by exploiting Lipschitz continuity of the trajectories{, which holds by \Cref{lem:lipschitz}}.
This shows uniqueness at $\linit$.
Let $\l$ be any equilibrium trajectory starting from $\linit$.
For any $\epsilon > 0$, we can think of $\l$ on $[\epsilon, \infty)$ as an equilibrium trajectory starting from $\l(\epsilon)$.
Since $\l(\epsilon) \to \linit$ as $\epsilon \to 0$, it follows from the above that $\sup_{\theta \geq 0}[\l(\theta + \epsilon) - \l^*(\theta)]$ converges to $0$ as $\epsilon \to 0$.
So $\l = \l^*$.

Note that while this conic view gives some nice intuition, our formal proof does not rely on this explicitly.
Instead, we show directly that not only does every trajectory reach steady state in some finite time, but that the time required to reach steady state scales with the distance between the initial point and $\linit$.

\medskip

This also shows continuity locally around $\linit$: equilibrium trajectories starting from small perturbations of $\linit$ remain close to the equilibrium trajectory starting from $\linit$.
Essentially, this shows continuity of the trajectories over a single phase, and it is not too difficult to deduce continuity of the entire equilibrium trajectory from this. 

\paragraph{Equilibria reach steady state.}
To prove \Cref{thm:uniqueness} and \Cref{thm:label-continuity}, the main remaining ingredient is to show that equilibria do always reach steady state. 
We prove the following theorem.

\begin{theorem}\label{thm:longterm}
	Consider a generalized network $G=(V,E)$ with free arcs $\Einf \subseteq E$.
      Let $\lambda$ be the thin flow direction for configuration $(E, \Einf)$, and let $\linit$ be any feasible labeling. 
        Then there is a value $T$ such that
        for any equilibrium trajectory $\l$ starting from $\linit$, $\l'_v(\theta) = \lambda_v$ for all $v \in V$ and $\theta \geq T$.
\end{theorem}

One can give explicit bounds on $T$ in terms of the instance and $\linit$; we do so in \Cref{sec:longterm}.
Note that in the theorem we show that after time $T$ \emph{all} labels change as {per the thin flow direction} $\lambda$. This is a stronger notion of steady state than asking only that queues change with the right rates as described in \Cref{sec:intro}. We will use the stronger notion as the definition of steady state. {We will also refer to $\lambda$ as \emph{the steady state direction}.} 

For uniqueness and continuity, the theorem is applied {to the current shortest path network}, with $\Einf = E^*_{\linit}$.
The choice $\Einf = \emptyset$ is of independent interest:
it corresponds to a real equilibrium trajectory on a given instance with no artificial free arcs.
As mentioned earlier, \CCO{} proved the following:
\begin{theorem}[\CCO]\label{thm:CCO}
	Consider an instance satisfying the following \emph{inflow condition}: the inflow rate $u_0$ does not exceed the minimum capacity of an $s$-$t$-cut in the network $G=(V,E)$.
	Then for any feasible initial condition $\linit$, after some finite time all queues remain constant; that is, there exists a $T$ such that $\l'_v(\theta) = 1$ for all $v \in V$ and $\theta \geq T$.
\end{theorem}
(Their theorem was stated for $\linit$ corresponding to the empty network, but their approach extends directly to arbitrary initial conditions.)
Our result for $\Einf = \emptyset$ thus removes the inflow condition, while modifying the notion of ``steady state'' appropriately. 

Proving \Cref{thm:longterm} constitutes the bulk of the technical work in this paper.
Our approach has its genesis in the proof of \Cref{thm:CCO} by \CCO, but generalizing their result is by no means straightforward.
It is not the introduction of free arcs that causes difficulty, but rather the violation of the inflow condition.
We will now give a high-level sketch of our approach, highlighting the main new difficulties and novelties compared to \CCO.  
The detailed proof can be found in \Cref{sec:longterm}.

\medskip

Let us first summarize the approach taken by \CCO{} under the inflow condition and $\Einf = \emptyset$.
They first pose and answer the following question (paraphrased):
\emph{which choices of initial condition $\linit$ have the property that $\l(\theta) = \linit + \lambda \theta$ is immediately an equilibrium trajectory?}
(In their case, the steady-state direction $\lambda$ is the all-ones vector.) 
They show that the answer is provided in full by considering the following primal-dual LP:
\begin{equation}\label{eq:CCOpd}
	\begin{aligned}
		\text{minimize} \quad &\sum_{e \in E} \tau_e f_e\\
		\text{s.t.} \qquad f\; &\text{ is an $s$-$t$-flow of value $u_0$}\\
		f_e &\leq \nu_e \phantom{0} \quad\text{for all } e \in E \\
		f_e &\geq 0 \phantom{\nu_e} \quad \text{for all } e \in E
	\end{aligned}\qquad \quad
	\begin{aligned}
		\text{maximize} \quad &u_0(d_t - d_s) - \sum_{e \in E} \nu_e p_e\\
		\text{s.t.} \qquad d_w -d_v -p_e &\leq \tau_e \phantom{0} \quad \text{for all } e=vw \in E\\
		p_e &\geq 0 \phantom{\tau_e} \quad \text{for all } e \in E
	\end{aligned}
\end{equation}
Dual optima $(d^*, p^*)$ are in one-to-one correspondence with steady-state initial conditions: $d^*$ represent a possible initial labeling,
and $p^*_e$ the queue length on arc $e$.
Primal optima, on the other hand, are in one-to-one correspondence with equilibria departure flows.
In other words: $(\l, x)$ in which $\l_v(\theta) = d^*_v + \theta$ for some dual optimum $d^*$, and $x'(\theta)$ is a primal optimum for almost every $\theta$, is an equilibrium; and all equilibria that are in steady state from the very beginning, are of this form.

It is by no means immediately apparent why answering this question regarding the characterization of steady states is helpful in proving convergence to steady state. 
The key novelty in \CCO{} is that the dual LP provides us with the correct potential function.
Namely, they define, given an equilibrium $(\l, x)$ with corresponding queueing delays $q$, the potential function
\begin{equation}\label{eq:ccopot}
	\Phi(\theta) \coloneqq u_0[\l_t(\theta) - \l_s(\theta)] - \sum_{e=vw \in E} \nu_e q_e(\theta).
\end{equation}
Then $\Phi(\theta)$ is the objective value of the feasible dual solution given by $d_v = \l_v(\theta)$ for all $v \in V$, and $p_e = q_e(\theta)$ for all $e \in E$ (feasibility being a consequence of feasibility of the labeling $\l(\theta)$).
The inflow condition ensures that the primal LP is feasible, and hence that the dual optimum has finite value; thus $\Phi(\theta)$ is bounded.
Moreover, $\Phi(\theta)$ turns out to be monotone\,---\,in fact, strictly monotone with slope bounded away from zero, until the point that steady state is reached.
The proof involves rewriting the derivative of $\Phi$, namely (for almost every $\theta$)
\begin{equation*}\label{eq:CCOderiv}
	\Phi'(\theta) = u_0[\l'_t(\theta) - \l'_s(\theta)] - \sum_{e \in E^*_\theta} \nu_e q'_e(\theta),
\end{equation*}
as an integral over a family of cuts, after which the inequality $\Phi'(\theta) \geq 0$ follows from the thin flow equations.
This shows convergence to steady state in finite time.

To generalize this result, we follow the same basic plan, but each stage presents new (and in some cases significant) additional challenges.
We will only consider $\Einf = \emptyset$ in this discussion, since as mentioned, this is not the major difficulty.
Characterizing steady state solutions is not much more difficult; replacing $\nu_e$ with $\hat{\nu}_e \coloneqq \nu_e\lambda_w$ for all $e=vw \in E$ in \eqref{eq:CCOpd} (where $\lambda$ is as defined in \Cref{thm:longterm}) does the job, in fact.
One can then attempt to define a potential based on the dual objective value in the same way, which would yield
\[ 
\Phi(\theta) = u_0[\l_t(\theta) - \l_s(\theta)] - \sum_{e = vw \in E} \lambda_w \nu_e q_e(\theta).
\]
This candidate potential is bounded, by feasibility of the primal (which is not difficult to show).
But unfortunately, it is \emph{not} monotone; an explicit counterexample can be found in~\cite{DarioThesis}.

It turns out that while a primal-dual LP characterizing the steady state is still the key to producing the correct potential, the situation is much more subtle.
The obvious generalization of \eqref{eq:CCOpd}, with $\nu_e$ replaced by $\hat{\nu}_e$ and no other changes, is not the correct one.
Rather, one must observe that there is a larger class of candidate LPs, from which a choice must be carefully made.
Let $y \in \Rplus^{\answer{E}}$ be such that $(y,\lambda)$ is a thin flow for configuration $(E, \emptyset)$. 
A first observation is that for arcs $e=vw$ with $\lambda_w > \lambda_v$, we may enforce the constraint $f_e = \hat{\nu}_e$,
and for arcs $e=vw$ with $\lambda_w < \lambda_v$, we may enforce the constraint $f_e = 0$, without changing the feasible set.
This comes from observing that if one looks at any set $S$ of the form $S = \{ v : \lambda_v \leq t\}$, for some $t \in \Rplus$, and this set is nontrivial (neither empty nor $V$), then $y_e = \lambda_w \nu_e$ for all $e=vw \in \delta^+(S)$ and $y_e = 0$ for all $e \in \delta^-(S)$ by the thin flow conditions. 

A followup observation, and the more crucial one, is that we have substantial flexibility in the primal objective function.
First, the coefficient in the objective for arcs $e=vw$ with $\lambda_w \neq \lambda_v$ has no impact on the set of optimal solutions, since the flows on these arcs are fixed.
Further, we can arbitrarily rescale (by positive values) the coefficients within each level set of $\lambda$; this again has no impact, because the flows entering and leaving the level set are determined.
In particular, we may replace the objective with $\sum_{e \in E} \hat{\tau}_e f_e$, where $\hat{\tau}_e \coloneqq \tau_e / \lambda_w$, which turns out to be the correct choice.
The mapping from $\l(\theta)$ to a feasible dual assignment also becomes more complicated, and is no longer based purely on the queue lengths.

%

\paragraph{Continuity revisited.}
We have so far discussed a notion of continuity that is very natural from our perspective of equilibria as trajectories $\l$.
But from the perspective of Nash flows over time, this notion perhaps does not lead itself to immediate interpretation.
There are in fact rather a large number of different questions about continuity that one could ask.

We will not attempt to exhaustively consider all possible or interesting notions of continuity in this work. 
Instead, we will discuss a few results that demonstrate that our ``primordial'' continuity result, with minimal effort, implies other forms of continuity.

Consider some instance, and the corresponding equilibrium trajectory $\l$ starting (for simplicity) from the empty network at time $0$.
Now suppose we perturb the transit times slightly, 
leading to a new equilibrium trajectory $\hat{\l}$ on the perturbed instance.
We show that $\hat{\l}$ remains similar to $\l$.
\begin{theorem}\label{thm:changetau}
	Let $G=(V,E)$ be a given network, with fixed capacities $\nu_e$ and inflow rate $u_0$, but variable transit times $\tau_e$.
	Let $\l^{(\tau)} \in \Traj$ be the trajectory corresponding to transit time vector $\tau$, starting from an empty network.
	Then $\tau \mapsto \l^{(\tau)}$ is continuous.
\end{theorem}

To see how this follows from our main continuity results, consider a single phase (continuity for the entire trajectory then follows by pasting this together appropriately, as with our main continuity result).
The steady state direction $\lambda$ does not depend on $\tau$ directly (only through the configuration that defines the phase).
The optimal solution and objective value to the primal-dual LP does change, but the optimal objective value, and hence the value that the potential $\Phi$ takes upon reaching steady state, is a continuous function of $\tau$.
The theorem then follows easily from the monotonicity of $\Phi$.

Similarly, we can consider perturbing the capacities $\nu_e$ and/or the inflow rate $u_0$.
Again we get continuity, though necessarily of a slightly weaker form: we must restrict ourselves to a compact interval.
This is because the steady state direction $\lambda$ \emph{does} change as $\nu$ or $u_0$ changes\,---\,but only continuously.
This means that once the final steady state of the overall equilibrium is reached, the trajectories may slowly diverge.
\begin{theorem}\label{thm:changenu}
	Let $G=(V,E)$, be a given network, with fixed transit times $\tau_e$ but variable inflow rate $u_0$ and capacities $\nu_e$.
	Let $\l^{(\nu, u_0)} \in \Traj$ be the trajectory corresponding to setting capacities to $\nu$ and inflow rate to $u_0$, starting from an empty network.
	Then $(\nu, u_0) \mapsto \l^{(\nu, u_0)}(\theta)$ is continuous for any fixed $\theta \in \Rplus$.
\end{theorem}

\paragraph{Single-commodity instances in general.}
In \Cref{sec:single_com}, we extend our results to general single-commodity instances.
Here there may be multiple origins, each with a fixed inflow rate, and multiple destinations; users are however indifferent between destinations, and will choose whichever destination they can reach first.

Handling multiple destinations is quite straightforward; simply contract together the destinations. 
Handling multiple origins is much less straightforward. 
It has been closely studied in \cite{sering2018multiterminal} (see also \cite{sering2020diss}), where equilibria are shown to exist and are described by a thin flow variant. 
We observe that an equilibrium in this model can also be seen as the concatenation of a finite number of equilibria of single-origin single-destination instances. This enables us to transfer all results form the single-source-single-destination setting to the general single-commodity setting.

	\section{Convergence to steady-state}\label{sec:longterm}

We are given a generalized network $G=(V,E)$ with free arcs $\Einf \subseteq E$. 
\nnote{Removed the following, which should no longer be needed given the updates: (For ease of notation we refrain from  using $\hat{G}=(V, \hat{E})$ as notation for a generalized network in what follows; the definitions for active and resetting arcs for generalized networks given in \Cref{eq:gen-actres} apply.)}
Let $\lambda$ denote the thin flow direction for configuration $(E, \Einf)$; we will refer to this as the \emph{steady-state direction}.

\begin{definition}\label{def:ss}
	Given an equilibrium trajectory $\l$, and $T \geq 0$, we say that $\l$ \emph{has reached steady state by time $T$} if $\l'(\theta) = \lambda$ for every $\theta \geq T$ where $\l$ is differentiable.
\end{definition}

A bound on the time to reach steady state must depend on the starting point $\linit$ of the trajectory.
To see this, just consider a single arc from $s$ to $t$ with sufficient capacity that will be non-resetting in steady state. 
The initial conditions can define an arbitrary long queue on that arc, which might need arbitrarily long time to deplete. 
We will provide a bound on the time to reach steady state that depends on the distance between $\linit$ and the \emph{steady-state set} $\Iss$, which we define next, and which goes to 0 as this distance goes to 0.

\begin{definition}
	The \emph{steady-state set} $\Iss$ denotes the set of feasible labelings such that an equilibrium trajectory starting from that initial point is immediately in steady state.
\end{definition}
That is, a trajectory reaches steady state precisely when it enters $\Iss$.

\begin{theorem}
	\label{thm:longterm-strong}
	Let $G=(V,E)$ with free arcs  $\Einf \subseteq E$ be a generalized network.
	There exists a constant $C_T$ depending only on the generalized network such that for any starting point $\linit { \in \Omega}$, 
	{an} equilibrium trajectory $\l$ starting from $\linit$ reaches steady state by time $C_T \cdot d(\linit, \Iss)$.
\end{theorem}
This clearly strengthens \Cref{thm:longterm}; the remainder of this section will constitute a proof of \Cref{thm:longterm-strong}.

Let $(\l,x)$ be any equilibrium starting from $\linit$, and let $q$ be the corresponding queueing delays.

\nnote{Moved the nonzero inflow rate lemma to Section 2.}

\paragraph{A primal-dual pair that almost characterizes steady state.}
We will begin by introducing a primal-dual pair of linear programs that will be hugely important. 
They are motivated by the primal-dual pair given by \CCO{} for the case $\Einf = \emptyset$, $\lambda \equiv 1$.
They will \emph{essentially} characterize the set of initial queue lengths for which the equilibrium is immediately in steady state.
The ``essentially'' here involves a technicality that we will come to in due course.

It will be useful to categorize non-free arcs according to their status with respect to $\lambda$.
Let
\begin{align*}
	E^> &\coloneqq \Set{ vw \in E \setminus \Einf | \lambda_w > \lambda_v},\\
	E^< &\coloneqq  \Set{vw \in E \setminus \Einf | \lambda_w < \lambda_v} \text{ and}\\
	E^= &\coloneqq  \Set{vw \in E \setminus \Einf | \lambda_w = \lambda_v}.
\end{align*}
Define 
\begin{equation}\label{eq:hatnutau}
	\hat{\nu}_e := \nu_e \cdot \lambda_w \qquad \text{and} \qquad \hat{\tau}_e := \tau_e / \lambda_w \qquad \text{for all } e=vw \in E.
\end{equation}

Consider the following minimum cost flow LP. Although we require it to be supported on $E \setminus E^<$, we deliberately include the flow variables for $E^<$. 
\begin{equation}\tag{P}\label{eq:primal}
	\begin{aligned}
		\text{minimize} \quad&\sum_{e \in E} \hat{\tau}_e f_e\\
		\text{s.t.} \qquad f \;&\text{ is an $s$-$t$-flow of value $u_0$}\\
		f_e &\leq \hat{\nu}_e \phantom{0}\qquad\text{for all } e \in E^= \cup E^< \\
		f_e &= \hat{\nu}_e  \phantom{0}\qquad\text{for all } e \in \Einf \cup E^>\\
		f_e &= 0 \phantom{\hat{\nu}_e}\qquad\text{for all } e \in E^<\\
		f &\geq 0.
	\end{aligned}
\end{equation}
Its dual (after some minor massaging) is 
\begin{equation}\tag{D}\label{eq:dual}
	\begin{aligned}
		\text{maximize} \quad &u_0(d_t - d_s) -\sum_{e \in E} \hat{\nu}_e p_e\\
		\text{s.t.} \qquad d_w -d_v -p_e &\leq \hat{\tau}_e\phantom{0}\qquad\text{for all } e=vw \in E \setminus E^<\\
		p_e &\geq 0 \phantom{\hat{\tau}_e}\qquad\text{for all } e \in E^= \cup E^<\\
	\end{aligned}
\end{equation}

We define the \emph{slack} of an arc $e \in E$ at some time $\theta$ by 
$s_e(\theta) \coloneqq [\tau_e + \l_v(\theta) - \l_w(\theta)]^+$ for $e \in E \setminus \Einf$, and $s_e(\theta) \coloneqq 0$ if $e \in \Einf$.
Note that active arcs have zero slack; it is a measure of how ``inactive'' an arc is.

Now make the following definitions, for any time $\theta$:
\begin{equation} \label{eq:nash_as_dual}
	\begin{aligned}
		d_v(\theta) &= \tfrac{\ell_v(\theta)}{\lambda_v}  \;\qquad\qquad\qquad\qquad\qquad\qquad {\text{for all }} v \in V\\
		p_e(\theta) &= \begin{cases} 
			\tfrac{\ell_w(\theta)}{\lambda_w} - \tfrac{\ell_v(\theta)}{\lambda_v} - \hat{\tau}_e + \tfrac{s_e(\theta)}{\lambda_w}\qquad &{\text{for all }} e=vw \in \Einf \cup E^>\\
			\tfrac{q_e(\theta)}{\lambda_w} &{\text{for all }} e=vw \in E^= \cup E^< .
		\end{cases}
	\end{aligned} 
\end{equation}

\begin{lemma}\label{lem:nash_as_dual}
	For any $\theta \geq 0$, $(d(\theta),p(\theta))$ is feasible for \eqref{eq:dual}.
\end{lemma}
\begin{proof}
	For the second dual constraint, we have that for any $e=vw \in E^= \cup E^<$, $p_e(\theta) = q_e(\theta)/\lambda_w \geq 0$.
	For the first dual constraint, consider any $e=vw \in E \setminus E^<$.
	If $e \in \Einf \cup E^>$, then 
	\[ d_w(\theta) - d_v(\theta) - p_e(\theta) = \hat{\tau}_e - s_e(\theta)/\lambda_w \leq \hat{\tau}_e. \]
	If instead $e \in E^=$, then 
	\[ d_w(\theta) - d_v(\theta) - p_e(\theta) = \tfrac1{\lambda_w}(\l_w(\theta) - \l_v(\theta) - q_e(\theta)) \leq \tfrac1{\lambda_w}\tau_e = \hat{\tau}_e. \]
\end{proof}

Let $\theta$ be a point of differentiability.
We can view $x'(\theta)$ as a primal solution to \eqref{eq:primal}; however, it will typically not be feasible.
If $(y,\lambda)$ is any solution to the thin flow equations for configuration $(E, \Einf)$, then it can easily be checked that $y$ \emph{is} feasible to the primal. 
We can ask the question: under what conditions is $x'(\theta)$ along with $(d(\theta), p(\theta))$ a feasible and \emph{optimal} primal-dual pair? 
We will give a sufficient condition in what follows; later it will turn out that this condition is in fact a precise characterization.

Call a node $v \in V$ \emph{active} at time $\theta$ if $v$ lies on some $s$-$t$-path in $E'_\theta$.
\begin{definition}
	We say that 
	$\l$ has the \emph{weak steady-state property} (or just \emph{WSS property}) at a point of differentiability $\theta$ if 
	\begin{inparaenum}[(i)] \item $\l'_v(\theta) = \lambda_v$ for every active node $v$, and \item $\l'_v(\theta) \geq \lambda_v$ for every inactive node $v$. \label{it:monotonicity}\end{inparaenum}
\end{definition}

The more crucial part of the definition is (i); (ii) will be convenient later for technical reasons.
If $\l$ has the WSS property for all $\theta' \geq \theta$, then $\l$ has reached steady state in the weaker sense alluded to in the introduction: all queues will grow at linear rate from time $\theta$ onwards. This follows as by \Cref{lem:labelsinc}, queueing arcs are always flow carrying. Moreover, for a flow carrying arc both end points are always active and thus $q_e'(\theta')= \l'_w(\theta') - \l'_v(\theta') = \lambda_w - \lambda_v$. 
For some intuition on the distinction between this and the stronger definition of reaching steady state given in \Cref{def:ss}, 
consider the example 
in \Cref{fig:wssnotss}.
This has the WSS property for all times.
Initially, the shortest path to $v$ (which is not active) is through $u$, and $\l'_v(\theta) = \l'_u(\theta) = 2$.
After the queue on $su$ has grown sufficiently large at some point $\theta'$ however, $sv$ becomes active, and $\l'_v(\theta')$ decreases to 1. 
At this point steady state has been reached.

\begin{figure}
	\centering
	\includegraphics[width=0.95\textwidth]{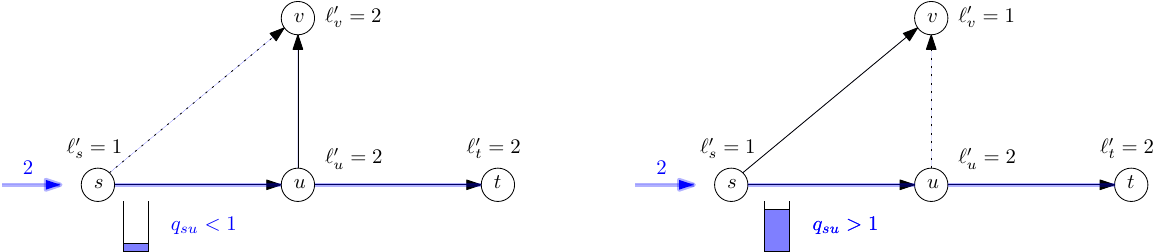}
	\caption{In the example the difference between the WSS property and steady state is depicted. Dotted arcs are inactive, solid arcs are active. All capacities are one and $3=\tau_{sv}>\tau_{su}+ \tau_{uv}=1+1$. 
		In the beginning, the label of $v$ is determined by $uv$, until $sv$ becomes active when the queue on $su$ reaches length 1.
		From this point, the label of $v$ is determined by $sv$ (and $uv$ becomes inactive). Since $v$ is not active the equilibrium already has the WSS property on the left, but is not yet in steady-state. Steady state is reached as soon as $sv$ becomes active.}
	\label{fig:wssnotss}
\end{figure}

\begin{lemma}
	\label{lem:WSS_properties}
	Suppose that $\l$ has the WSS property at a point of differentiability  $\theta$.
	Then the following holds:
	\begin{enumerate}[(i)]
		\item \label{it:arcs_have_good_type}
		$e \in E^*_\theta$\/ for all $e \in E^>$, and
		$x'_e(\theta) = 0$ for all $e \in E^<$.
		
		\item \label{it:WSS_and_opt} $x'(\theta)$ is an optimal solution to \eqref{eq:primal}, and $(d(\theta), p(\theta))$ is an optimal solution to \eqref{eq:dual}. 
	\end{enumerate}
\end{lemma}
\begin{proof}
	We fix $\theta$ and omit it if convenient, i.e., we write $x_e' $ instead of $x_e'(\theta)$, $\l'_v$ for $\l'_v(\theta)$ and so on.
	\begin{enumerate}[(i)]
		\item
		Consider the cut $Q_z = \{ v \in V \mid  \lambda_v < z \}$ for any $z$.
		Note that all arcs that leave $Q_z$ are in $\Einf \cup E^>$ and all arcs that enter $Q_z$ are in $\Einf \cup E^<$.
		Thus, on all leaving arcs $e=vw$ as well as on the entering arcs $e=vw \in \Einf$ , we have $y_e = \nu_e \lambda_w$, while $y_e=0$ on $e\in E^<$. 
		Now, consider the flow $x'$ on arcs crossing the cut. On all arcs leaving $Q_z$ we have either $0=x_e' <y_e$ or
		$x'_e \leq \l'_w \nu_e  = \lambda_w \nu_e =y_e$. 
		On the other hand, on the arcs entering $Q_z$ we have $x_e' \geq 0 = y_e$ for $e \in E^<$ and $x'_e =\l'_w \nu_e = \lambda_w \nu_e =y_e$ for $e \in \Einf$, by the WSS property, using that every such arc is always resetting and hence lies on an active $s$-$t$-path.
		Since $x'$ and $y$ are both $s$-$t$-flows {of value $u_0$} and hence have the same net flow across $Q_z$, all inequalities must hold with equality. In particular $x_e' = 0 $ for every $e \in E^<$ crossing $Q_z$, and every arc $e \in E^>$ that crosses $Q_z$ is flow carrying in $x'$.
		In the latter case, $v$ and $w$ both certainly lie on an active $s$-$t$-path, and so $\l'_w > \l'_v$ by the WSS property, meaning that (since we are at a point of differentiability) $e \in E^*_\theta$.
		Thus, we showed the properties we want for all arcs crossing some $Q_z$.
		But every arc in $E^>$ and $E^<$ crosses $Q_z$ for some choice of $z$, and so the claim follows.
		\item 
		Dual feasibility has already been observed in \Cref{lem:nash_as_dual}.
		Primal feasibility follows straightforwardly from (\ref{it:arcs_have_good_type}):
		\begin{itemize}
			\item For arcs $e \in E^<$ we have $x_e'=0$ as desired. 
			\item For arcs $e=vw \in E^> \cup \Einf$ we have that $e \in E^*_\theta$ and so $x_e' = \l_w' \nu_e = \lambda_w \nu_e =\hat\nu_e$ (using that $x'_e > 0$ by \Cref{lem:labelsinc}). 
			\item 
			For arcs $e = vw \in E^=$, either $x_e'=0$, or if $x'_e > 0$ (so that $\l'_w = \lambda_w$), then $x_e' \leq \l_w' \nu_e = \lambda_w \nu_e =\hat\nu_e$. In either case, $0 \leq x_e '\leq \hat{\nu}_e$.
		\end{itemize}
		
		It remains to check that the complementary slackness conditions are satisfied.
		These are 
		\begin{subequations}
			\renewcommand{\theequation}{\theparentequation.\arabic{equation}}
			\begin{align}
				f_e = 0 \quad&\text{or}\quad p_e = {d_w - d_v - \hat{\tau}_e} && \text{for all } e \in E \setminus E^< \label{eq:cs1}\\
				f_e = {\hat{\nu}}_e \quad&\text{or}\quad p_e = 0 && \text{for all } e \in E^= \cup E^<. \label{eq:cs2}
			\end{align}
		\end{subequations}
		These again follow rather easily from (\ref{it:arcs_have_good_type}) and \eqref{eq:nash_as_dual}.
		\begin{itemize}
			\item \emph{Condition \eqref{eq:cs1}}:
			For $e\in \Einf$, we defined $s_e=0$ and thus $p_e = {d_w - d_v - \hat{\tau}_e}$.
			For $e \in E^> \cup E^=$ the second part of condition holds if the arc is active, and $x'_e = 0$ if it is inactive.
			
			\item \emph{Condition \eqref{eq:cs2}}:
			For $e \in E^=$, if $p_e >0$ this implies that $e \in E^*_\theta$ and thus {that it is} flow carrying (\Cref{lem:labelsinc}). This means $x_e' =\l'_w \nu_e =\lambda_w \nu_e =\hat\nu_e$.
			For $e \in E^<$ we always have $p_e=0$ by \eqref{it:arcs_have_good_type}.
		\end{itemize}
	\end{enumerate}
\end{proof}

\paragraph{A potential function.}
Let
\begin{equation} \label{eq:potential_function}
	\Phi(\theta) \coloneqq u_0(d_t(\theta) - d_s(\theta)) - \sum_{e=vw \in E} \hat{\nu}_e p_e(\theta).
\end{equation}
Note that if $\Einf = \emptyset$ and $\lambda \equiv 1$, then $\Phi(\theta) = u_0(\l_t(\theta) - \l_s(\theta)) -\sum_{e \in E} \nu_e q_e(\theta)$ exactly matches the potential used in \CCO.

Boundedness is the easier claim.
Since the primal \eqref{eq:primal} has a feasible solution $y$, and the dual is feasible as well, the primal-dual pair has a finite objective value, which we denote by $\OPT$.
\begin{lemma}\label{lem:bound}
	$\Phi(\theta) \leq \OPT$.
\end{lemma}
\begin{proof}
	Immediate, since $(d(\theta), p(\theta))$ is a feasible solution to the dual by \Cref{lem:nash_as_dual}, and $\Phi(\theta)$ is precisely the objective value of this dual solution. 
\end{proof}

\newcommand{\xx}{\dot{x}}

Now we come to monotonicity. 
\begin{lemma}\label{lem:phi-monotone}
	At every point of differentiability $\theta$, $\Phi'(\theta) \geq 0$, and if $\Phi'(\theta) = 0$, then $\l$ satisfies the WSS property at time $\theta$.
\end{lemma}
\begin{proof}
	It will be convenient to work with the graph $H=(V, E'_\theta \cup \{ts\})$, where we define the rescaled capacity of $ts$ to be $\hat{\nu}_{ts} = u_0$. We will write $\delta^-(S)$ and $\delta^+(S)$ for the set of arcs entering and leaving a set of vertices $S$ in $H$, respectively.
	Let $\xx$ be the circulation on $H$ obtained by setting $\xx_{ts} = u_0$ and $\xx_e = x'_e(\theta)$ for all $e \in E'_\theta$.
	The arc $ts$ behaves similarly to arcs in $\Einf$, as we take it to always be active and resetting; define $s_{ts}(\theta)=0$ for all $\theta$.
	We will omit the parameter $\theta$ and just write $\Phi'$, $\l'_v$, $q'_e$ and $s'_e$, whenever it is clear.
	
	For any $z \geq 0$, let
	\[
	S_z := \Set{ v \in V | \frac{\l'_v}{\lambda_v} < z}.
	\]
	Also define $\Epot := E^*_\theta \cup (E^> \cap E'_\theta) \cup \{ts\}$. 
	\begin{claim}\label{clm:pot-cut}
		\[ \Phi' \geq - \int_0^\infty \hat{\nu}(\delta^+(S_z) \cap \Epot) - \hat{\nu}(\delta^-(S_z) \cap \Epot) \diff z. \]
	\end{claim}
	\begin{nestedproof}
		Differentiating $\Phi$, we have
		\[
		\Phi' = -\sum_{e=vw \in \Einf \cup E^> \cup \{ts\}} \hat{\nu}_e \left(\tfrac{\l'_w}{\lambda_w} - \tfrac{\l'_v}{\lambda_v} + \tfrac{s'_e}{\lambda_w}\right) - \sum_{e=vw \in E^= \cup E^<} \nu_e q'_e. 
		\]
		We claim that $\Phi' \geq -  \sum_{e=vw \in \Epot} \hat{\nu}_e \cdot \left(\frac{\l'_w}{\lambda_w} - \frac{\l'_v}{\lambda_v}\right)$.
		This follows from comparing terms.
		{ For arcs in $E^< \setminus E^*_\theta$ and $E^= \setminus E^*_{\theta}$, the contributions to both terms are zero. For the arc $ts$ and arcs in  $\Einf$, $E^> \cap E'_\theta$ and $E^= \cap E^*_{\theta}$, the contributions to both terms are identical.}
		
		For $e=vw \in E^> \setminus E'_\theta$, we have 
		\[ -\hat{\nu}_e\left(\tfrac{\l'_w}{\lambda_w} - \tfrac{\l'_v}{\lambda_v}\right) - \nu_e s'_e = \hat{\nu}_e \l'_v\left(\tfrac1{\lambda_v} - \tfrac1{\lambda_w}\right) \geq 0. \]
		For arcs $e=vw \in E^< \cap E^*_\theta$, we have
		\[
		-\nu_e q'_e = -\hat{\nu}_e\left(\tfrac{\l'_w}{\lambda_w} - \tfrac{\l'_v}{\lambda_w} \right) \geq -\hat{\nu}_e \left(\tfrac{\l'_w}{\lambda_w} - \tfrac{\l'_v}{\lambda_v} \right).
		\]
		Finally, we observe that
		\begin{align*}
			\sum_{e=vw \in \Epot} \hat{\nu}_e \cdot \left(\frac{\l'_w}{\lambda_w} - \frac{\l'_v}{\lambda_v}\right) &=
			\int_0^\infty \sum_{\substack{e=vw \in \Epot:\\ \l'_v/\lambda_v < z \leq \l'_w/\lambda_w}} \kern-5mm \hat{\nu}_e 
			\;\;\;\;-\;\; \sum_{\substack{e=vw \in \Epot:\\ \l'_w/\lambda_w < z \leq \l'_v/\lambda_v}}\kern-5mm \hat{\nu}_e \;\;\diff z\\
			&= \int_0^\infty \hat{\nu}(\delta^+(S_z) \cap \Epot) - \hat{\nu}(\delta^-(S_z) \cap \Epot) \diff z. \qedhere
		\end{align*}
	\end{nestedproof}

	\begin{claim}\label{clm:crossingSz}
		Fix any $z$.
		\begin{enumerate}[(i)]
			\item 	\label{it:all_flow_counted}$\delta^-(S_z) \cap E'_\theta \subseteq \Epot$.        	
			\item If  $e=vw \in \delta^+(S_z) \cap \Epot$, then $\xx_e = \l'_w \nu_e$.
		\end{enumerate}
	\end{claim}
	\begin{nestedproof}
		\begin{enumerate}[(i)]
			\item Consider any arc $e=vw \in \delta^-(S_z) \cap E'_\theta$.
			Either $e \in E^*_\theta \subset \Epot$ or (using that $\theta$ is a point of differentiability) $e \in E'_\theta$ with $\ell'_w = \ell'_v$.
			In the latter case, we must have $\lambda_w > \lambda_v$ (since $\ell'_w / \lambda_w < z \leq \ell'_v / \lambda_v$).
			So $e \in \Epot$. 
			\item 
			If $e \in E^*_\theta \cup \{ts\}$ then the claim trivially follows.
			Otherwise, $e \in E^> \cap \delta^+(S_z)$, i.e.,
			$\lambda_w> \lambda_v$ and $\frac{\l_w'}{\lambda_w} > \frac{\l_v'}{\lambda_v}$, 
			which implies that $\l'_w > \l'_v$. 
			Since $e \in E'_\theta$, the claim follows from the thin flow conditions. 
		\end{enumerate}
	\end{nestedproof}
	
	Since $\xx$ is a circulation and {by \eqref{it:all_flow_counted}}, it follows that the incoming flow on arcs in $\Epot$ is not smaller than the outgoing flow on arcs in $\Epot$, i.e., 
	\begin{equation}\label{eq:circbound}
		\xx(\delta^-(S_z) \cap \Epot) \geq \xx(\delta^+(S_z) \cap \Epot).
	\end{equation}
	
	For all active arcs, we certainly have $\xx_e \leq \l'_w \nu_e$.
	Together with the claim, \eqref{eq:circbound} yields the following sequence of inequalities.
	\begin{align*}
		z \cdot \hat{\nu}(\delta^-(S_z) \cap \Epot) &\geq \sum_{e=vw \in \delta^-(S_z) \cap \Epot} \frac{\l'_w}{\lambda_w} \cdot \hat{\nu}_e \\
		&= \sum_{e=vw \in \delta^-(S_z) \cap \Epot} \l'_w \cdot \nu_e\\
		&\geq \xx(\delta^-(S_z) \cap \Epot)\\
		&\geq \xx(\delta^+(S_z) \cap \Epot)\\
		&= \sum_{e=vw \in \delta^+(S_z) \cap \Epot} \l'_w \cdot \nu_e\\
		&= \sum_{e=vw \in \delta^+(S_z) \cap \Epot} \frac{\l'_w}{\lambda_w} \cdot \hat{\nu}_e\\
		&\geq z \cdot \hat{\nu}(\delta^+(S_z) \cap \Epot).
	\end{align*}
	By \Cref{clm:pot-cut}, $\Phi' \geq 0$.
	
	\medskip
	
	Finally, suppose that $\Phi'(\theta) = 0$. 
	Then for almost all values of $z$, all inequalities in the above sequence hold with equality.
	This in particular implies that for all $e=vw \in \Epot$, $\l'_w/\lambda_w = \l'_v/\lambda_v$. 
	For if not, we would have either that the first inequality is strict for all $z \in (\l'_w/\lambda_w, \l'_v/\lambda_v]$ (if $\l'_w / \lambda_w < \l'_v / \lambda_v$), or that the last inequality is strict for all $z \in [\l'_v/\lambda_v, \l'_w/\lambda_w)$ (if $\l'_w / \lambda_w > \l'_v / \lambda_v$).

	We claim that any active arc $e=vw \in E'_\theta$ satisfies $\l'_w / \lambda_w \geq \l'_v / \lambda_v$.
	For if not, then choosing any $z \in (\l'_w/\lambda_w, \l'_v/\lambda_v)$ and applying \Cref{clm:crossingSz}, $e \in \Epot$, which we have just shown is impossible. 
	Since $\l'_s = \lambda_s = 1$, and every node is reachable from $s$ via an active path, $\l'_v \geq\lambda_v$ for all $v$.
	
	Now consider any $s$-$t$-path $P$ in $E'_\theta$.
	Then any arc $e=vw$ in $P$ satisfies $\l'_w/\lambda_w \geq \l'_v / \lambda_v$, and also since $ts \in \Epot$, $\l'_t/\lambda_t = \l'_s/\lambda_s = 1$.
	So $\l'_v/\lambda_v = 1$ for all $v \in P$.
	Since this holds for any $s$-$t$-path in $E'_\theta$, $\l'_v = \lambda_v$ for all active nodes.
	Thus the WSS property is satisfied.
\end{proof}

\begin{lemma}
	\label{lem:wss_reached}
	There is a constant $C_1$, depending only on the network (more concrete on the graph and the capacities) but not on $\linit$, such that $\l$ has the WSS property for all $\theta \geq C_1(\OPT - \Phi(0))$.
\end{lemma}	
\begin{proof}
	By \Cref{lem:phi-monotone}, $\Phi$ is nondecreasing, and strictly increasing at time $\theta$ unless the WSS property is satisfied at time $\theta$.
	Since the WSS property implies optimality by \Cref{lem:WSS_properties}, we deduce that $\Phi$ must be strictly increasing until the moment that $\Phi(\theta) = \OPT$, at which point the WSS property is satisfied for all larger times. 
	
	To bound this time (and see that it is finite), {n}ote that $\Phi'(\theta)$ takes on only a finite set of values depending only on the network, since its value is fully determined by the configuration $(E'_\theta, E^*_\theta)$. More concretely, for every possible valid configuration of active and resetting edges the $\l'$ values are determined by the thin flow conditions and thus only depend on the capacities.
	So we can choose $C_1$ depending only on the network such that $1/C_1 \leq \Phi'(\theta)$ for all points of differentiability $\theta$ where $\Phi'(\theta) > 0$.
	So $\Phi(\theta) = \OPT$ for all $\theta \geq C_1(\OPT - \Phi(0))$, as desired.
\end{proof}
We remark that it follows from this that if $\Phi(\theta) = \OPT$, then $\l$ satisfies the WSS property.  So while we stated \Cref{lem:WSS_properties}~(\ref{it:WSS_and_opt}) as an implication, the WSS property is in fact a characterization of optimality.

\paragraph{The final phases.}
We have shown that for all $\theta$ large enough, $\l$ satisfies the WSS property, and in particular, $\l'_v(\theta) = \lambda_v$ for all active nodes $v$.
It remains to show that, possibly after some further time has passed, $\l'_v(\theta) = \lambda_v$ for \emph{all} nodes $v$.

\medskip
{
We will need the following technical lemma.
\begin{lemma}
    \label{lem:lambda_sensible}
    \answer{Assume $\ell$ has the WSS property at time $\theta$. Then every}
    $v \in V$ is reachable in $(V, E^=)$ from an active node.
\end{lemma}
\begin{proof}
    Let $S$ be the set of inactive nodes that cannot be reached in $(V,E^=)$ from an active node, and suppose for a contradiction that $S \neq \emptyset$. 
    There are no arcs in $E^=$ entering $S$ by definition.
    By \Cref{lem:WSS_properties} and \Cref{lem:labelsinc}, there can be no arcs of $E^>$ entering or leaving $S$ (since any such arcs have $x'_e(\theta) > 0$, and hence have active endpoints). \answer{Similarly, by \Cref{lem:labelsinc} no arc of $\Einf$ enters or leaves $S$.}    


    
     Let $y$ be a flow so that $(\lambda, y)$ is a full solution to the thin flow equations for configuration $(E, \Einf)$. 
     Recall that $\lambda$  is the thin flow direction and thus $(\lambda, y)$ is a maximal solution. 
     We will show that this leads to a contradiction. 
    Let $\bar{\lambda} = \lambda + \epsilon\chi(S)$, for $\epsilon > 0$ small enough that $\bar{\lambda}_z \geq \bar{\lambda}_v$ for all $zv \in \delta^-(S) \cup \delta^+(S)$.
    It is possible to choose such an $\epsilon$, since $S$ has no incoming arcs in $E^= \cup E^>\cup \Einf$ and no outgoing arcs in $E^>\cup \Einf$.
Since $y_e = 0$ on all arcs entering $S$, leaving $S$, or within $S$, it is easy to check that $\bar{\lambda}$ still satisfies \eqref{eq:tf-group}, a contradiction.
\end{proof}
}

Let $(\gamma_v(\theta))_{v \in V}$ be the shortest path labels from $s$ in the network $(V, E \setminus E^<)$ with costs given by the slacks $s_e(\theta)$ (recall that $s_e(\theta) = [\tau_e + \l_v(\theta) - \l_w(\theta)]^+$ for $e=vw \in E \setminus \Einf$, and $s_e(\theta)=0$ for $e\in \Einf$).
\begin{lemma}\label{lem:gammazero}
	Consider any point of differentiability $\theta$ for which $\Phi(\theta) = \OPT$.
	Then for any $v$, $\gamma_v(\theta) = 0$ if and only if $\l'_v(\theta) = \lambda{_v}$.
\end{lemma}
\begin{proof}
	
Since $\Phi(\theta)=\OPT$, $\Phi'(\theta) = 0$, and so by \Cref{lem:phi-monotone} the WSS property is satisfied. 	
		
		Consider any active node $v \neq s$, and let $e=uv$ be any active arc entering $v$.
		Then $\l'_v(\theta) = \lambda_v$ and $\l'_u(\theta) = \lambda_u$ by the WSS property.
		Since $e$ is active and $\theta$ is a point of differentiability, either $\l'_v(\theta) = \l'_u(\theta)$, or $e \in E^*_\theta$ (meaning $x'_e(\theta) > 0)$ and so $e \notin E^<$ by \Cref{lem:WSS_properties}.
		Either way, $e  \in E \setminus E^<$. {Hence, there is an active $s$-$v$-path of arcs in $E \setminus E^<$, i.e., a path without slack.}
		It follows that $\gamma_v(\theta) = 0$ for all active nodes $v$.
		
		Consider now any inactive node $v$. 
		Let $P$ be a path of active arcs { to $v$} starting from an active node $u$, and with all other nodes of $P$ being inactive.
		If there are multiple choices for $P$, choose one with the fewest number of arcs in $E^<$.
		No arc of $P$ is resetting (since  (end) nodes of resetting arcs are always active).
		By \Cref{lem:WSS_properties}, this further implies that no arc of $P$ is in $E^> \cup \Einf$, so $\lambda_u \geq \lambda_v$.
		Since $\theta$ is a point of differentiability, it also follows that $\l'_u(\theta) = \l'_v(\theta)$.
		
		If $\gamma_v(\theta) = 0$, then $P$ contains only arcs of $E^=$ (since there must be an active path from $s$ to $v$ in $(V, E \setminus E^<)$, and hence such a path from the set of active nodes, and $P \subseteq (E^< \cup E^=$)\answer{)}. So $\l'_v(\theta) = \lambda_v$.
		Conversely, if $\l'_v(\theta) = \lambda_v$, so that $\lambda_u = \l'_u(\theta) = \l'_v(\theta) = \lambda_v$, then $P \subseteq E^=$, and $\gamma_v(\theta) = \gamma_u(\theta) = 0$.
\end{proof}

Let $\Delta \coloneqq 1 / \min \set{\abs{\lambda_v - \lambda_w} : \lambda_v \neq \lambda_w }$ (or $\Delta \coloneqq 0$ if the minimum is over an empty set). 
\begin{lemma} \label{lem:post_phase}
	Suppose that $\Phi(T_1)=\OPT$ and let $T_2 \coloneqq \Delta \cdot \max_{v \in V} \gamma_v(T_1)$.
	Then $\l$ reaches steady state by time $T_1 + T_2$.
\end{lemma}
\begin{proof}
	Let $\theta \geq T_1$ be a point of differentiability and let $S \coloneqq \set{v \in V : \l'_v(\theta) = \lambda_v}$.
	
	We show that for any $v \in V \setminus S$, $\gamma'_v(\theta) \leq -\Delta^{-1} < 0$ (or if $\Delta=0$, then $S=V$).
	This suffices: after time $T_2 = \Delta \cdot \max_{v \in V} \gamma_v(T_1)$ has passed beyond time $T_1$, the slack $\gamma_v(\theta)$ has decreased to $0$ for every $v \in V$, and hence all nodes are in $S$.
	
	Consider any $v$ with $\gamma_v(\theta) > 0$; note that $v \notin S$ by \Cref{lem:gammazero}.
        {By \Cref{lem:lambda_sensible}, $v$ is reachable from $S$ in $(V, E^=)$.
	So choose $P_1$ to be a $u$-$v$-path in $E^= \setminus E[S]$ with $u \in S$ chosen to have minimum possible slack, and moreover, 
    where this remains true for some small interval of time.}
	That is, $\sum_{e \in P_1} s_e(\xi) = \gamma_v(\xi)$ for all $\xi \in [\theta, \theta + \epsilon)$ for some small $\epsilon > 0$.
	This is possible given that $\theta$ is a point of differentiability; we can choose $\epsilon$ so that all slacks are linear in this interval.

	Let $P_2$ be a $w$-$v$ path in $E'_\theta \setminus E[S]$ with $w \in S$. 
	We have that $\lambda_v = \lambda_{u} = \l'_{u}$ and $\lambda_v<\lambda_{w} = \l'_{w}=\l'_v$.
	The inequality follows since $P_2$ must use at least one arc from $E^<$, or otherwise $P_2$ would be a no slack path in $E\setminus E^<$.
	This observation further implies that if $\Delta=0$, meaning $E^< = \emptyset$, then $S=V$.
	
	As there are no queues on the arcs on $P_1$ we have
	\[\gamma'_v(\theta) = \sum_{e \in P_1} s'_e(\theta) = \sum_{v'w' \in P_1} \l'_{v'}(\theta)- \l'_{w'}(\theta) =\l'_{u} -\l'_v = \lambda_v - \lambda_{w} \leq -\Delta^{-1}. 
	\]
	This completes the proof.
\end{proof}

\paragraph*{Putting it all together.}
Altogether, we can now prove \Cref{thm:longterm-strong}.

\begin{proof}[Proof of \Cref{thm:longterm-strong}]
	Let $R := d(\linit, \Iss)$.
        {We can view the potential $\Phi$ as a function of $l \in \Omega$, by using $l$ rather than $\l(\theta)$ everywhere in its definition. This function is Lipschitz continuous; denote the Lipschitz constant by $\Cpot$.}
        \nnote{Technically $\Phi$ is never defined as a function of $l$, only of $\theta$, so perhaps confusing. Adjusted.}
	Then $\OPT - \Phi(0) \leq \Cpot R$.
	Let $T_1 = C_1(\OPT - \Phi(0))$, with $C_1$ as guaranteed by \Cref{lem:wss_reached}, so that
	$\l$ has the weak steady-state property from time $T_1$ onwards. 
	Then $T_1 \leq C_1' R$, where $C_1' = C_1 \Cpot$. 
	
	By \Cref{lem:post_phase}, $\l$ reaches steady state by time $T_1+ T_2$, where $T_2 = \Delta \cdot \max_{v \in V}{\gamma_v(T_1)}$.
	To bound $T_2$, we observe that $\gamma_v$, viewed as a function of the labeling, is Lipschitz continuous with Lipschitz constant $2|V|$, since the slacks are Lipschitz continuous with Lipschitz constant $2$.
	Thus $\gamma_v(0) \leq 2R|V|$.
	Since $\|\l(T_1) - \l(0)\| \leq \kappa T_1$, this time by Lipschitz continuity of the trajectory, we deduce the very crude bound $\gamma_v(T_1) \leq 2|V|(R + \kappa T_1)$ for all $v$.
	
	Using the above bound for $T_1$, we obtain 
	$T_2 \leq 2\Delta|V| R (1 + \kappa C_1')$, 
	and hence by choosing $C_T := C_1' + 2\Delta |V|(1 + \kappa C_1')$, $T_1 + T_2 \leq C_T \cdot R$.
\end{proof}

	\section{Uniqueness and continuity}\label{sec:unique}
In this section, we prove the uniqueness and continuity of equilibrium trajectories.
Our main tool will be the following lemma showing continuity for some small interval.
\begin{lemma}[Local continuity]
\label{lem:local_continuity}
Consider a generalized network $G=(V,E)$ with free arcs $\Einf \subseteq E$.
Fix a feasible labeling $\linit$, and let $\l^*$ be defined by $\l^*(\theta) \coloneqq \linit + \theta \cdot X(\linit)$.
Then there exists $\eta > 0$ such that the following holds.
For any sequence $\l^{(1)}, \l^{(2)}, \ldots$ of equilibrium trajectories where $\l^{(i)}(0) \to \linit$ as $i \to \infty$, 
$\l^{(i)}(\theta) \to \l^*(\theta)$ as $i \to \infty$ for all $\theta \in [0, \eta]$.
\end{lemma}

\begin{proof}
	Since there exists {an $\eta>0$ and} $\delta > 0$ such that for any equilibrium trajectory $\l$ starting within $B_{\delta}(\linit)$, 
	$E^*_{\linit} \subseteq E^*_{\l(\theta)}$ and $E'_{\l(\theta)} \subseteq E'_{\linit}$ for all $\theta \in [0, \eta]$, we can focus on the local network $\hat G$, where only currently active arcs are present and currently resetting arcs become free arcs. 
	So set $\Eloc \coloneqq E'_{\linit}$ and $\hat{E}^\infty \coloneqq E^*_{\linit}$, and
	define the local network to be $\Gloc = (V, \Eloc)$ with free arcs $\hat{E}^\infty$.
	Observe that the set of feasible labelings of the new and old generalized network coincide on $B_{\delta}(\linit)$.	
	
	Thus, we can consider the local network $\Gloc$, where $C_T$ is as defined in \Cref{thm:longterm-strong} for $\Gloc$ and $\kappa$ is the Lipschitz constant of any equilibrium trajectory.
    Since $\linit$ is in the steady-state set of the local network, by applying \Cref{thm:longterm-strong}, we get 
    \[
        \nnorm{\l^* (\theta) - \l^{(i)}(\theta)} \leq 2\kappa \nnorm{\linit - \l^{(i)}(0)} C_T \quad  \text{for all } \theta \in [0,\eta] \text{ and } i.
    \] 
    Here, we exploited that once $\l^{(i)}$ reaches steady state, it moves in parallel to $\l^*$.
    Thus $\l^{(i)}(\theta) \to \ell^*(\theta)$ as $i \to \infty$, as desired.
    \end{proof}

\paragraph{Uniqueness.}
\answer{We are now ready to prove uniqueness of equilibrium trajectories. }
%
%
\begin{proof}[\answer{Proof of \Cref{thm:uniqueness}}]
Assume for contradiction that there are two distinct equilibrium trajectories $\l^{(1)}$ and $\l^{(2)}$ starting from $\linit$. 
Without loss of generality we assume that the trajectories diverge right from the start: otherwise, simply consider the moment in time when the two trajectories diverge, and treat this as the initial condition (shifting times appropriately).

{Consider the sequence of equilibrium trajectories $\l^{(1)}, \l^{(2)}, \l^{(1)}, \l^{(2)} , \ldots$ and apply \Cref{lem:local_continuity} to this sequence. Then there exists an $\eta > 0$ and a trajectory $\l^*$, so that $\l^{(1)}(\theta)=\l^{(2)}(\theta)=\l^*(\theta)$ for $\theta \in [0, \eta]$. This contradicts the assumption that the trajectories $\l^{(1)}$ and $\l^{(2)}$ diverge at time 0.}
%
\end{proof}

\paragraph{Continuity.}
\answer{We now prove continuity of equilibrium trajectories.}
\begin{proof}[\answer{Proof of \Cref{thm:label-continuity}}]



\answer{We begin by proving a weaker continuity statement.}
For every $\theta \in \Rplus$, define $\answer{\Psi_\theta \colon \Omega \to \R^V}$ by $\answer{\Psi}_\theta(\linit) = [\answer{\Psi}(\linit)](\theta)$; that is, we map an initial condition to the value of the resulting equilibrium trajectory at a fixed time $\theta$.
\begin{claim}\label{lem:continuity-fixedtheta}
    For any $\theta \in \Rplus$, $\answer{\Psi}_\theta$ is continuous.
\end{claim}
\begin{nestedproof}
Suppose that is not the case; then the following set is bounded:
\[
    M \coloneqq \Set{\vartheta \in \Rplus | \answer{\Psi}_{\theta'} \text{ is continuous for all } \theta' \in [0, \vartheta]}.
\]
Let $\xi \coloneqq \sup M$ (note that $M \neq \emptyset$ as $0 \in M$). 
As a first step we show that $\xi \in M$. 
Recall that all equilibrium trajectories are $\ldmax$-Lipschitz.
This proves that $\answer{\Psi}_\xi$ is continuous, because for every $\epsilon > 0$ we can find a $\delta > 0$ such that
{for $\theta' = \xi - \frac{\epsilon}{3\ldmax}$ }
and $\hatlinit \in B_{\delta}(\linit) \cap \answer{\Omega}$, it holds that
\[
    \nnorm{\answer{\Psi}_{\xi}(\linit) - \answer{\Psi}_{\xi}(\hatlinit)} \leq \nnorm{\l(\xi) - \l(\theta')} + \nnorm{\answer{\Psi}_{\theta'}(\linit) - \answer{\Psi}_{\theta'}(\hatlinit)}+ \nnorm{\hat\l(\theta') - \hat\l(\xi)} \leq \tfrac{\epsilon}{3} + \tfrac{\epsilon}{3} + \tfrac{\epsilon}{3} \leq \epsilon.
\]
Here, $\l$ and $\hat \l$ are the equilibrium trajectories starting with $\linit$ and $\hatlinit$, respectively. 

We can consider, for every equilibrium trajectory $\l$,
$\l(\xi)$ as the initial condition of the equilibrium trajectory $\theta' \mapsto \l(\theta' + \xi)$.
Applying \Cref{lem:local_continuity}, we obtain that for a small duration $[\xi, \xi + \epsilon]$ the equilibrium trajectory $\l$ depends continuously on the value of $\l(\xi)$.
But since $\answer{\Psi}_\xi$ is continuous, meaning that $\l(\xi)$ depends continuously on $\linit$, this shows that $\answer{\Psi}_{\xi + \epsilon}$ is also continuous, contradicting our choice of $\xi$.
\end{nestedproof}


Given an interval $I \subseteq \Rplus$, let $\answer{\Psi_I(l)}$ be the restriction of $\answer{\Psi(l)}$ to the time interval $I$, for any $\answer{l} \in \Omega$.
It is a basic fact that a sequence of Lipschitz continuous functions {with a common Lipschitz constant} that converge pointwise on a compact interval converge uniformly.
Thus $\answer{\Psi}_I$ is continuous for any compact interval $I$.
To extend to the non-compact interval $\Rplus$, we again make use of our steady state result.
By \Cref{thm:longterm-strong} there is a constant $T$ such that every equilibrium trajectory with starting point in $B_{\delta}(\linit) \cap \Omega$ for some fixed $\delta>0$ has reached steady state by time $T$ (note that $T$ depends on distance of $\linit$ to steady state and $\delta$; more concretely we can choose $T = C_T (d(\linit, I_{ss})+\delta)$). 
Thus for any $\answer{l} \in B_{\delta}(\linit) \cap \Omega$ and $\theta > T$, $\answer{\Psi}_\theta(\answer{l}) = \answer{\Psi}_T(\answer{l}) + (\theta - T)\cdot \lambda$,
where $\lambda$ is the steady-state direction, i.e., $\lambda$ is the thin flow direction for configuration $(E, \Einf)$. 
Thus 
$\sup_{\theta \geq 0}\|\answer{\Psi}_\theta(\answer{l}) - \answer{\Psi}_\theta(\linit)\| = \sup_{0 \leq \theta \leq T} \|\answer{\Psi}_\theta(\answer{l}) - \answer{\Psi}_\theta(\linit)\|, $
and continuity of $\answer{\Psi}$ follows.
\end{proof}

{\paragraph{Piecewise-constant inflow rates.}
The results on continuity \answer{ over a compact interval}  and uniqueness translate easily to the setting of piecewise-constant inflow rates. 
Simply apply the continuity and uniqueness results to each maximal interval for which the inflow is constant; since the composition of continuous functions is continuous, there are no difficulties.
More formally, local continuity as stated in \Cref{lem:local_continuity} still holds, where the inflow rate is taken to be the constant inflow over some small interval of time starting from a time at interest, and so \Cref{thm:uniqueness} as well as \Cref{lem:continuity-fixedtheta} (continuity over a compact interval) still follow.
For a constant inflow, we were able to show continuity over the unbounded interval $\Rplus$; 
this relied on the convergence to steady-state, which certainly does not hold in general for a time-varying inflow.
So this result does not extend, unless the inflow is constant after some finite time.
}

	\section{Continuity more generally}\label{sec:gen-continuity}
In this section, we show how continuity with respect to the travel times or the capacities can be fairly easily deduced from the continuity of equilibrium trajectories in the initial labeling.

\paragraph{Continuity with respect to $\tau$.}
We first show \Cref{thm:changetau}, that equilibrium trajectories are continuous with respect to the transit times of the instance.
So fix all aspects of the instance apart from the transit time vector~$\tau$.
We will use $\Omega^\tau$ to denote the set of feasible labelings for the instance corresponding to $\tau$, and similarly $X^\tau$ for the vector field.

An important first observation is that for a fixed configuration $(E', E^*)$, the thin flow equations have no dependence on $\tau$ whatsoever. 

First, we prove local continuity, in the sense of \Cref{lem:local_continuity}. 
\begin{lemma}
\label{lem:local_continuity_tau}
Fix $\tauinit \in \Rplus^E$ with no directed cycles of zero cost, as well as $\linit \in \Omega^\tauinit$.
Let $\l^*$ be defined by $\l^*(\theta) = \linit + \theta \cdot X^\tauinit(\linit)$.
Then there exists $\epsilon > 0$ such that the following holds.
Consider any sequence $\tau^{(1)}, \tau^{(2)}, \ldots$ of valid transit time vectors converging to $\tauinit$, 
and any corresponding sequence of equilibrium trajectories $\l^{(1)}, \l^{(2)}, \ldots$ with $\l^{(i)}$ being in the instance with transit time vector $\tau^{(i)}$ and $\l^{(i)}(0) \to \linit$ as $i \to \infty$.
Then $\l^{(i)}(\theta) \to \l^*(\theta)$ as $i \to \infty$ for all $\theta \in [0, \epsilon]$.
\end{lemma}
\begin{proof}
    As in the proof of \Cref{lem:local_continuity}, we pick $\epsilon$ small enough so that we can restrict our attention to the local network at $\linit$.
    More precisely, suppose $\epsilon > 0$ is small enough that for all $i$ sufficiently large, the active arcs of $\l^{(i)}(\theta)$ are in $E'_{\linit}$ for all $\theta \in [0, \epsilon]$, and all arcs of $E^*_{\linit}$ are resetting arcs of $\l^{(i)}(\theta)$.
    As before let $(\lambda, y)$ be a solution to the thin flow equations for the configuration $(E'_{\linit}, E^*_{\linit})$ considering transit times $\tauinit$. 
Note that the steady state direction $\lambda$ is independent of transit times.

Let $\Phi^{(i)}$ be the potential function as defined in \eqref{eq:potential_function} for $\l^{(i)}$ in the network $G^{\tau^{(i)}}$, 
and let $\Phinom$ be the potential for $\l^*$ in $G^{\tauinit}$. 
Let moreover $\OPT^{(i)}$ be the upper bound on $\Phi^{(i)}$ and $\OPT^\circ$ the upper bound on $\Phinom$. The dual feasibility set and thus the optimal value depend continuously on $\tau$ and thus $\OPT^{(i)} \to \OPT^\circ$ as $i \to \infty$.
Since the potential at time $0$ depends continuously on transit times and the initial values, we have $\Phi^{(i)}(0) \to \Phinom(0)$ as $i \to \infty$. 
%
By \Cref{lem:wss_reached}, the time $T_1^{(i)}$ after which $\l^{(i)}$ satisfies the weak steady state properties is upper bounded by $C_1(\OPT^{(i)} - \Phi^{(i)}(0))$, where $C_1$ depends only on the graph and the capacities and not $\tau$. 
As $\Phinom(0) = \OPT^{\circ}$, we obtain
 $T_1^{(i)} \to 0$ as $i \to \infty$.

Next we show that also the duration after $T_1^{(i)}$ until steady state $T_2^{(i)}$ goes to $0$ as $i \to \infty$.
By \Cref{lem:post_phase}, we have $T_2^{(i)} \leq \Delta \max_{v \in V}  \gamma^{(i)}_v(T_1^{(i)})$. 
The definition of $\Delta$ depends only on the steady-state direction $\lambda$ and is thus independent of the transit times. 
Since $T_1^{(i)} \to 0$, $\l^{(i)}(T_1^{(i)}) \to \linit$; in particular, the distance between $\l^{(i)}(T_1^{(i)})$ and the steady-state set is going to $0$.
By \Cref{lem:gammazero} and Lipschitz continuity of each $\gamma^{(i)}_v$ with respect to labels, it follows that $\gamma^{(i)}_v(T^{(i)}_1) \to 0$ for all $v$.
 
Thus $T^{(i)} := T_1^{(i)}+ T_2^{(i)} \to 0$, and we conclude that $\|\l^*(\theta) - \l^{(i)}(\theta)\| \leq 2\ldmax T^{(i)} + \|\linit - \l^{(i)}(0)\| \to 0$ as $i \to \infty$, for any $\theta \in [0, \epsilon]$. 
\end{proof}

Given this, the conversion of local continuity to full continuity over the whole trajectory is essentially identical to the argument in \Cref{sec:unique}.
The extension of local continuity to continuity over any compact interval is the same.
To extend this to the entire trajectory, we just again use that the steady-state direction, 
obtained as the solution to the thin flow equations with configuration $(E, \emptyset)$ is the same.
We can obtain a uniform bound on the time to reach this steady state for all choices of $\tau$ in some neighborhood, and so the claim follows.


\paragraph{Continuity with respect to $\nu$.}
Next, we consider perturbing the capacities $\nu_e$ and/or the inflow rate $u_0$.
Perturbing the inflow rate can be thought of as perturbing an arc capacity, by inserting a dummy arc $s's$ with capacity $u_0$, setting $s'$ to be the new source, and choosing the inflow of the new instance to be large; perturbing the capacity of $s's$ is then functionally equivalent to perturbing the inflow of the original instance.
So we will consider only perturbations of $\nu$.

Define, for any $\nu \in \Rplusplus^E$, $G^\nu$ and $X^\nu$ for the instance and vector field corresponding to $\nu$.
This time, the set of feasible labelings $\Omega$ does not depend on $\nu$.

The new ingredient compared to perturbing $\tau$ is that the steady-state direction $\lambda^\nu$ 
\emph{does} depend on $\nu$.
However, it does so continuously.
\begin{lemma}\label{lem:lambda-continuous}
    Fix some valid configuration $(E', E^*)$, and let $\lambda^\nu$ be the thin flow direction for this configuration in $G^\nu$, for any capacity vector $\nu$. \nnote{Put back 'thin flow direction'.} 
    Then $\lambda^\nu$ depends continuously on $\nu$. 
    \nnote{in $\Rplus^E$ perhaps? Formally we asked for $\nu_e > 0$, but that's just because a 0-capacity edge is the same as a missing edge. There is mention of some set being closed below, that's the only reason I bring this up. For now I'm just removing the explicit mention of the domain of $\nu$.}
    \laura{oh indeed. good point. I would be ok with leaving this as it is, as it seems to be just a formality.  
    	}
\end{lemma}
\begin{proof}
    Consider an ordered partition $\mathcal{P} = (V_1, V_2, \ldots, V_k)$ of $V$, which we can view as an assignment $\pi: V \to \mathbb{N}$ 
    that labels each node with the index of the part it lies in. Associated with this, define the following linear system:
    \begin{align*}
        \l'_s &= 1,\\
        \l'_w &= \l'_v \;\;\qquad\qquad \text{ for all } v,w \text{ where } \pi(w) = \pi(v),\\
        x'_e &= \nu_e\l'_w \;\;\quad\qquad \text{ for all } e=vw \text{ where } e \in E^* \text{ or } \pi(w) > \pi(v),\\
        x'_e &= 0 \quad\qquad\qquad \text{ for all } e=vw \text{ where } e \notin E^* \text{ and } \pi(w) < \pi(v),\\
        \sum_{e \in \delta^+(v)} x'_e - \sum_{e \in \delta^-(v)}x'_e &= \begin{cases} u_0 \;\;\qquad&\text{ if } v=s,\\
            0 &\text{ if } v \in V \setminus \{s,t\}.
        \end{cases}
    \end{align*}
    This linear system may have multiple solutions. If there exists a solution $(\l', x')$ for which $0 \leq x'_e \leq \nu_e\l'_w$ for all $e=vw$, then it is easy to verify that this satisfies the thin flow conditions; essentially, we have guessed the ordering between label derivatives. 
    Say that $\mathcal{P}$ is a \emph{correct (ordered) partition} if this holds, and let $N[\mathcal{P}]$ denote the set of capacities for which the linear system for this partition is correct. 
    Note that $N[\mathcal{P}]$ is a closed set. \todo{Why?}
    Existence of thin flows implies that for every $\nu$, there is at least one ordered partition $\mathcal{P}$ for which $N[\mathcal{P}]$ contains  $\nu$ (there may be more than one; it is possible that for some partition $(V_1, \ldots, V_k)$, we obtain a solution with $\l'_v = \l'_w$ for $v \in V_j$, $w \in V_{j+1}$, in which case the partition obtained by merging $V_j$ and $V_{j+1}$ would also be correct).
    Uniqueness of thin flows implies that for any correct partition $\mathcal{P}$, and \emph{any} solution $(\l', x')$ to the corresponding system, $\l' = \lambda^\nu$.

    It follows that to show continuity of $\lambda^\nu$, it suffices to show, for a fixed ordered partition $\mathcal{P}$, continuity of the solution $\l'$ of the linear system with respect to $\nu$ within $N[\mathcal{P}]$.
    Substitute out the $x'$ variables, to reduce to a linear system $A^\nu\l' = b$.
    The entries of $A^\nu$ clearly depend continuously (indeed, linearly) on $\nu$. \todo{Clarity of why this is preserved after substituting?}
    Further, since $\l'$ is uniquely determined in $N[\mathcal{P}]$, $A^\nu$ is nonsingular in this set.
    Hence $\nu \mapsto (A^\nu)^{-1}$ is also continuous in $N[\mathcal{P}]$.
\end{proof}

Since $\Omega$ does not depend on $\nu$, we simplify our life and maintain a fixed initial condition (of course, the result could be combined with the results on continuity with respect to initial conditions, if desired).
The local continuity statement now becomes much simpler, since we know that the trajectory will move in the thin flow direction for some configuration, and so becomes an immediate consequence of the previous lemma.
\begin{lemma}
\label{lem:local_continuity_nu}
Fix $\nuinit \in \Rplusplus^E$, as well as $\linit \in \Omega$.
Let $\l^*$ be defined by $\l^*(\theta) = \linit + \theta \cdot X^{\nuinit}(\linit)$.
Then there exists $\epsilon > 0$ such that the following holds.
Defining $\l^{\nu}$ to be the equilibrium trajectory for $G^{\nu}$ with $\l^{\nu}(0) = \linit$, 
$\l^{\nu}(\theta) \to \l^*(\theta)$ as $\nu \to \nuinit$ for all $\theta \in [0, \epsilon]$.
\end{lemma}
\begin{proof}
    Let $\lambda^\nu$ be the solution to the thin flow equations with configuration $(E'_{\linit}, E^*_{\linit})$ in $G^\nu$, for any capacities $\nu$.
    By \Cref{lem:lambda-continuous}, $\lambda^\nu$ depends continuously on $\nu$.
    For any $\nu$, there exists some $\epsilon(\nu) > 0$ so that $\l^{\nu}(\theta) = \linit + \theta X^\nu(\linit)$ for $\theta \in [0, \epsilon(\nu)]$ (this is from considering the equilibrium constructed by the $\alpha$-extension procedure~\cite{koch2011nash,cominetti2015existence}, which we have shown is the only possible equilibrium). 
    Further, if we restrict our attention to some sufficiently small ball around $\nuinit$, we can choose $\epsilon(\nu) = \epsilon$ for some fixed $\epsilon > 0$ for all $\nu$ in this ball. \todo{Handwavy. (Not a priority though.)}
    The claim thus follows from the previous lemma.
\end{proof}
By Lipschitz continuity, we deduce (precisely as in \Cref{sec:unique}) that the equilibrium trajectory is continuous with respect to $\nu$ for any compact interval.
This time, however, we cannot go beyond this, since now the steady-state direction can depend on $\nu$, meaning that trajectories may slowly diverge as time goes on.
So we cannot get uniform convergence of the trajectories for all time.

	\section{Single-commodity instances with multiple origins and destinations}
\label{sec:single_com}

All results presented in the paper can be extended to a setting with a single commodity but multiple sources and sinks. Here, we are given a set of sources, each with a fixed inflow rate and agents enter the network via these sources and then travel to which ever destination they can reach fastest. 

First observe that this setting can be straightforwardly reduced to the setting with multiple sources but only one sink. 
Since the agents are indifferent regarding which sink they arrive at, we can simply contract all the sinks together into a single sink.
So we assume just a single sink for the remainder.

Unfortunately, we cannot handle multiple sources so easily.
To see why the situation is different here, consider a setting with two sources $s_1$ and $s_2$ both connected to some vertex $v$, which is then connected to the sink $t$. 
Let us assume that the travel time from $s_1$ to $v$ is much shorter than the travel time from $s_2$ to $v$. 
In this case the particles which enter the network via $s_2$ at time 0 interact with particles which enter the network via $s_1$ at a much later time. 
If we were to contract $s_1$ and $s_2$ together, no particles would depart from $s_2$ at time $0$, which is certainly incorrect.

Fortunately, the question of how to handle multiple source has already been discussed by Sering and Skutella~\cite{sering2018multiterminal}.
They showed that Nash flows over time exist in this setting, and show that their label derivatives satisfy a modified version of the thin flow equations almost everywhere.
{Their formalization defines labels $\l_v(\theta)$, but since there is no common source, it is not possible to identify $\theta$ as a departure time from a source, and indeed it can no longer be interpreted as a time at all. 
    Nonetheless, $\l_v(\theta)$ retains its interpretation as the earliest arrival time for a particle indexed by $\theta$, and they show that an equilibrium is described by the usual requirement that all flow is supported on $E'_\theta$ (which is defined as before, via \eqref{eq:actres} and $E'_\theta := E'_{\ell(\theta)}$).
    For each source $s$ and time $\xi$ there is a $\theta$ for which $\l_s(\theta) = \xi$, and so this ensures that from every source, all particles traverse quickest paths to the sink. 
    They also show that the label derivatives satisfy a variant of the thin flow equations.
}

We cannot quite use their results in a black box fashion, since we would like to simply apply our results for the single-source single-sink case, rather then redoing the proofs with slightly different thin flow equations. 

\newcommand{\actsources}[1]{\mathcal{S}_{#1}}
\newcommand{\sources}{\mathcal{S}}
\newcommand{\zeroset}[1]{\mathcal{Z}_{#1}}
\newcommand{\inactivity}{\beta}

Let $\sources = \{s_1, s_2, \ldots, s_m\}$ be the set of sources (and $t$ the sink, as usual), and let $u_i$ be the inflow rate at $s_i$.
It will be convenient to define the following network $\hat{G}$. 
Its node set $V(\hat{G})$ consists of $V$ along with an additional ``super-source'' $s$.
Its arc set $E(\hat{G})$ consists of $E$, along with the arcs $E_s := \{ ss_i: s_i \in \sources\}$.
All arcs in $E$ maintain their transit times and capacities; 
arc $ss_i$ has capacity $u_i$ and transit time $0$ for each $s_i \in \sources$.
Let $u_0 \coloneqq \sum_{i=1}^m u_i$.

We will need to consider the details of the modified thin flow equations of \cite[Definition 3]{sering2018multiterminal}.
{We will consider maximal solutions to the following, which, as $E_{\theta}'$ is acyclic, can easily  be seen to be equivalent to  the modified thin flow equations of \cite[Definition 3]{sering2018multiterminal}.}
\begin{equation}\label{eq:modified-tf-ss}
	\begin{alignedat}{2}
		x' &\text{ is a static $s$-$t$-flow of value $u_0$ supported on $E'_\theta \cup E_s$} \\
		\l'_{s_i} &= x'_{ss_i} / u_i &&\text{ for all } s_i \in \sources\\ 
		\l'_w &
		\mathbin{{\leq}}
		\rho_e(\l'_v, x'_e)  && {\text{ for all } e = vw \in E_{\theta}'},\\
		\l'_w &= \rho_e(\l'_v, x'_e) &&\text{ for all } e = vw \in E'_\theta \text{ with } x'_e > 0.
	\end{alignedat}
\end{equation}
(Recall the definition of $\rho_e$ from \eqref{eq:tf-group}, and the characterization of $E'_\theta$ from \eqref{eq:actres}, which remains valid here).


Call a source $s_i$ \emph{active} at $\theta$ if there is a path in $(V, E'_\theta)$ from $s_i$ to $t$.
Let $\actsources{\theta}$ is the set of active sources at time $\theta$.
Now let $\zeroset{\theta}$ be the set of nodes reachable from the \emph{inactive} sources $\sources \setminus \actsources{\theta}$ along a path of active arcs.
Then it is easy to observe that in the modified thin flow equations, an inactive source $s_i$ must have $x'_{ss_i} = \l'_{s_i} = 0$, since there is no flow through $s_i$ supported on $E'_\theta${.}
It also then follows that all nodes in $v \in \zeroset{\theta}$ satisfy $\l'_v = 0$.
Let $E''_\theta = E'_\theta \cup \{ ss_i : s_i \in \actsources{\theta}\}$, and let $E^{**}_\theta := E^*_\theta \cup \{ ss_i: s_i \in \actsources{\theta}\}$.
Then the following is equivalent to \eqref{eq:modified-tf-ss}:
\begin{equation}\label{eq:modified-tf}
	\begin{alignedat}{2}
		x' &\text{ is a static $s$-$t$-flow of value $u_0$  supported on $E''_\theta$} \\
		\l'_w &
		\mathbin{{\leq}}\rho_e(\l'_v, x'_e) \quad && \hspace{-2em} {\text{ for all } e = vw \in E''_\theta } \text{ with } w \in V \setminus \zeroset{\theta}\\
		\l'_w &= \rho_e(\l'_v, x'_e) &&\hspace{-2em} \text{ for all } e = vw \in E''_\theta \text{ with } x'_e > 0 {\text{ and } w \in V \setminus \zeroset{\theta}}\\
		\l'_{v} &= 0 &&\hspace{-2em} \text{ for all } v \in \zeroset{\theta}.
	\end{alignedat}
\end{equation}
Here, we have extended $\rho_e$ to $e=ss_i$ for $s_i \in \actsources{\theta}$ in the obvious way $\rho_e(\l'_s, x'_e) = x'_e / u_i$, given that these are all considered resetting arcs.
Now define the generalized network $G_\theta$ by first removing $\zeroset{\theta}$ from $\hat{G}$, and then choosing the set of free arcs to be $E^\infty_\theta := \{ ss_i : s_i \in \actsources{\theta}\}$.
{Note that $(E(G_\theta), E^\infty_\theta)$ is a valid configuration, since no arc $ss_i$ lies on a directed cycle.}
\nnote{Adjusted this, and I removed any reference to the definition of a feasible labeling. I think this is OK?}\laura{I think yes, One thing still confueses me. First, I thought $\l_{s_i}(\theta)=0$, was a condition for a feasible labeling, but we deleted all inactive sources already at this point. (And for the sources which are present, we get that they are active as $ss_i$ is resetting and the definition of a valid configuration, thus guarantees an active path to $t$.) So what is the sentence: We need one further restriction: if $s_i$ is {an} inactive source at time $\theta$, then $\l_{s_i}(\theta) = 0$.... saying? }
\nnote{I agree it was weirdly phrased. I assume we just meant that we need to fix these to some reasonable value (and then because we set $\l'_v = 0$ in the dead part, these are then naturally fixed as well. I adjusted the text, OK?}
Then \eqref{eq:modified-tf} is identical to the usual thin flow equations \eqref{eq:tf-group} on $G_\theta$, along with $\l'_v = 0$ for $v \in \zeroset{\theta}$.
{We also fix $\l_{s_i}(\theta)=0$ for all inactive sources at time $\theta$ (these labels are not meaningful); we include this as a requirement for a labeling to be feasible.}
As usual, we can consider equilibrium trajectories starting from any feasible labeling $\linit$.

\medskip
Given some equilibrium trajectory $\l$, remember that the slack of an arc is defined as $s_e(\theta)=[\l_v (\theta) + \tau_e- \l_w(\theta) ]^+$ for $e=vw \in E \setminus \Einf$ and $s_e(\theta)=0$ for $e \in \Einf$. 
For any source $s_i \in \sources$, let $\inactivity_{s_i}$ be the minimal cost of an $s_i$-$t$ path, where arc costs are defined by the slacks:
\[
\inactivity_{s_i}(\theta) = \min \Bigl\{ \sum_{e \in P} s_e(\theta) : P \text{ is an $s_i$-$t$-path}\Bigr\}.
\]
\begin{lemma}
	\label{lem:actsources_behave_well}
	There is a constant $\ldmin>0$, depending only on the network, such that the following holds:
	\begin{compactenum}[(i)]
		\item\label{it:inactivity_dec}
		For all points of differentiability $\theta$ and $s_i \in \sources \setminus \actsources{\theta}$, {we have} $\inactivity_{s_i}'(\theta)\leq -\ldmin$. 
		\item The set of active sources $\actsources{\theta}$ is increasing in $\theta$, and
		$\actsources{\theta} = \sources$ for all $\theta \geq\sum_{e \in E}\tau_e/ \ldmin$.
	\end{compactenum}
\end{lemma}
\begin{proof}
	\begin{enumerate}[(i)]
		\item
		Fix $\theta$ and $s_i$. 
		Let $P$ be any $s_i$-$t$-path of minimum slack at time $\theta$.
		Let $\tilde{P}$ be the maximal prefix of $P$ such that no arc of $\tilde{P}$ has a queue.
		The last node $u$ of $\tilde{P}$ must be connected to $t$ via a path of active arcs. 
		So for any $e \in P \setminus \tilde{P}$, $s'_e(\theta) = 0$.
		Let $\tilde\inactivity := \sum_{e \in \tilde{P}} s'_e(\theta) = \sum_{e \in P} s'_e(\theta)$.
		Clearly $\inactivity_{s_i}'(\theta) \leq \tilde\inactivity$, so it suffices to upper bound $\tilde\inactivity$.
		
		For any $e=vw \in \tilde{P}$, since $e$ has no queue, we have $s'_e(\theta) = \l_v'(\theta) - \l_w'(\theta)$ (recall that $\theta$ is a point of differentiability).
		Define $\ldmin$ to be the minimum possible value of $\l'_v$ for any node $v$ and any thin flow solution $\l'$ in $G_\theta$, over all possible valid configurations (this is strictly positive by \Cref{lem:labelsinc} and since no node of $\zeroset{\theta}$ is in $G_{\theta}$). 
		Then 
		\[ \tilde\inactivity = \sum_{e =vw \in \tilde P} (\l'_v(\theta) - \l'_w(\theta)) = \l'_{s_i}(\theta) - \l'_{u}(\theta) = -\l'_u (\theta) \leq - \ldmin.
		\]
		
		\item 
		Note that a source $s_i$ is active when $\inactivity_{s_i}(\theta)=0$.  It follows directly that the set of active sources $\actsources{\theta}$ only increases in $\theta$, since $\inactivity_{s_i}$ is  continuous{, nonnegative and monotone}.
		
			Next, observe that for every inactive source $s_i$ the maximal slack on a path to the sink is upper bounded by $\sum_{e \in E} \tau_e$.
			Let $P$ be an $s_i$-$t$-path of minimum slack at time $0$ and let $\tilde{P}$ be a maximal prefix of $P$ such that no arc of $\tilde{P}$ has a queue. Then (similarly to the above)
			\[ 
			\inactivity_{s_i}(0) = \sum_{e \in P} s_e(0) = \sum_{e=vw \in \tilde{P}} (\tau_e + \l_v(0) - \l_w(0)) = \l_{s_i}(0) - \l_u(0) + \sum_{e \in \tilde{P}} \tau_e,
			\]
			where $u$ is the last node of $\tilde{P}$. 
			Since $\l_{s_i}(0) = 0$, it follows that $\beta_{s_i}(0) \leq \sum_{e \in \tilde{P}} \tau_e$. 
			Applying (\ref{it:inactivity_dec}), the claim follows.
	\end{enumerate}
\end{proof}

Stronger versions of the following theorem, analogous to \Cref{thm:longterm-strong}, are certainly possible; we give the most straightforward and concise version here.
\begin{theorem}
	In the general single-commodity setting every equilibrium trajectory reaches steady state. 
\end{theorem}
\begin{proof}
	By \Cref{lem:actsources_behave_well}, we have that for all $ \theta \geq \sum_{e \in E} \tau_e/\ldmin$ all sources are active and remain active. Thus $\zeroset{\theta}= \emptyset$, and so after some initial period, the trajectory is described as an equilibrium trajectory of a single, unchanging single-source single-sink instance.
	\Cref{thm:longterm} can thus be directly applied.
\end{proof}

As in the single-source-single-sink setting we show a local form of continuity. 

\begin{lemma}
	\label{lem:local_continuity_general}
	Let $\l$ be the equilibrium trajectory starting from some initial feasible labeling $\linit$, and let $\l^{(1)}, \l^{(2)}, \ldots$ be a sequence of equilibrium trajectories  
	where $\l^{(i)}(0) \to \linit$ as $i \to \infty$.
	Then there exists $\eta > 0$ such that 
	$\l^{(i)}(\theta) \to \l(\theta)$ as $i \to \infty$ for all $\theta \in [0, \eta]$.
\end{lemma}
\begin{proof}
	In the following, $\actsources{\theta}$, $\zeroset{\theta}$ etc.\ refer to the trajectory $\l$.
	We choose $\eta$ small enough that during the interval $[0, \eta]$, no new source will become active, and the set $\zeroset{\theta}$ does not change. 
	To guarantee this, consider the minimal slack $\tilde{s}$ (with respect to $\l$) of an arc leaving $\zeroset{0}$, and choose $\eta = \tilde{s}/(4\kappa)$. 
	Then by Lipschitz continuity, any trajectory starting within distance $\tilde{s}/4$ of $\linit$ will remain within distance $\tilde{s}/2$ of $\linit$ throughout $[0, \eta]$.
	Since $\l^{(i)}(0) \to \linit$, for any $i$ large enough we have that $\l^{(i)}(\theta)$ remains within distance $\tilde{s}/2$ of $\linit$, and hence that all arcs leaving $\zeroset{0}$ still have slack. 
	This implies in particular that the set of active sources for $\ell^{(i)}$ is a subset of the active sources of $\ell$, for all $i$ large enough.
	
	There may be some sources that are active for $\l$ but inactive for $\l^{(i)}$ at time $0$.
	We first argue that all such sources will become active very quickly.
	For suppose $s_i$ is such a source. 
	We have that $\inactivity^{(i)}_{s_i}(0) \leq \abs{V} (\l^{(i)}(0) - \linit ) $, where $\inactivity^{(i)}$ is defined with respect to $\ell^{(i)}$.
	By \Cref{lem:actsources_behave_well}, this implies that $s_i$ becomes active by time $T^{(i)} := (\l^{(i)}(0) - \linit )  \abs{V}/\ldmin$.
	
	Thus, from time $T^{(i)}$ both $\l^{(i)}$ and $\l$ share the same set of active sources and the same set of nodes reachable from active sources, as long as $i$ is sufficiently large.
	So $\l$ and each $\l^{(i)}$ (for $i$ large enough) can all be viewed as equilibrium trajectories in the same single-source single-sink instance; for $\l$ this is true for the interval $[0, \eta]$, and for $\l^{(i)}$ on the interval $[T^{(i)}, \eta]$.
	Further, $T^{(i)} \to 0$, meaning that $\l^{(i)}(T^{(i)}) \to \l(T^{(i)})$. 
	Applying \Cref{thm:label-continuity} to this common instance, we deduce that $\l^{(i)}$ converges to $\l$ on the interval $[0,\eta)$.
	
	\nnote{Maybe this is skipping a bit too much in places.}
\end{proof}

Local continuity implies both uniqueness and global continuity.
The argument is the same as in the single-source-single-sink setting (see \Cref{sec:unique}, in particular proof of \Cref{thm:uniqueness,thm:label-continuity}), and we do not repeat it here.

\begin{corollary}
	In the general single-commodity setting, there is a unique equilibrium trajectory for any given initial labeling, and this trajectory depends continuously on the initial labeling.
\end{corollary}

	\section{Conclusion}\label{sec:conclusion}

In this paper we have demonstrated continuity in several different forms. 
As discussed in the introduction, continuity is a crucial property for the model to have if its behavior is to have any relevance to real-world traffic situations. 
However, it is also crucial to achieve stronger forms of continuity and stability.

Consider a variant of the fluid queueing model where each agent is not infinitesimally small, but controls a ``packet'' of small but finite size. Is it true that the equilibrium of this packet model converges to the equilibrium of the fluid queueing model, as the packet size goes to $0$?

While intuitive, such a result was not known.
It is clearly important: a negative answer would bring into question the relevance of the study of equilibria in the fluid queuing model.
Empirical evidence has been provided~\cite{ZiemkeEtAl2020FlowsOverTimeAsLimitOfMATSim}, and rigorous results have been obtained for the static model~\cite{cominetti2020convergence}.
As some first steps, Sering, Vargas Koch and Ziemke~\cite{sering2021convergence} showed convergence of the underlying packet model to the fluid queueing model, but in a weak sense that does not consider equilibrium behavior.

Our continuity result can be viewed as a necessary step to answering this question, but it does not suffice.
We can show continuity with respect to a bounded number of ``bumps'' of the equilibrium, but equilibria of a packet model version of an instance deviate from the equilibrium of the fluid queueing model in a much more substantial way.

Fortunately, in followup work~\cite{OSV23}, we have been able to demonstrate a convergence result of this form.
Our result is quite general, and also shows (for example) that $\epsilon$-approximate equilibria converge to exact dynamic equilibria as $\epsilon \to 0$, and that small time-varying perturbations to transit times (see \cite{pham2020dynamic}) do not impact the equilibrium behavior very much.
The results of this paper (continuity as well as long-term behavior) are absolutely crucial ingredients in the new result.


\medskip

An important question that our results touch on concerns whether equilibria have a \emph{finite} number of phases.
Our results on long-term behavior imply that after finite time, there are no further phases.
They also rule out points of accumulation ``from the right'': for any given $\theta$, there exists an $\epsilon > 0$ so that the time interval $(\theta, \theta + \epsilon)$ lies within a single phase.
What our results do not rule out is the possibility of accumulation points ``from the left'':
that is, a moment $\theta$ such that there are an infinite number of phases in $(\theta-\epsilon, \theta)$ for any positive $\epsilon$.
More geometrically, we rule out the outward-spiraling situation of \Cref{fig:spiraling}, but not an inward-spiraling trajectory.
It remains an open question whether our methods can be strengthened to answer this question.

\medskip


    As already discussed, our results on uniqueness and continuity carry over to piecewise-constant time-varying inflows.
We expect that the claims hold for arbitrary (appropriately well-behaved) inflow functions, but we do not know how to extend the arguments of this paper to show this, and we leave this as an open question.

\medskip

Our results show that for any given instance, there is a function $\delta(\epsilon)$ that goes to zero as $\epsilon$ goes to zero such that if two equilibrium trajectories start within distance $\epsilon$, then they remain within distance $\delta(\epsilon)$.
A weakness is that the explicit dependence of $\delta$ on $\epsilon$ is not very well controlled (and we do not describe it explicitly). 
It seems plausible that the correct dependence should even be \emph{linear}. 
We leave this as a direction for future work.

Finally, a major shortcoming of our results is that they apply only to single-commodity instances.
Our understanding of the multi-commodity version of the model is much less developed in general; existence of equilibria is known~\CCL), and \cite{sering2018multiterminal,sering2020diss} discuss some structural aspects. {Interestingly, in \cite{iryo2011multiple} it was shown that equilibria need not be unique anymore in multi-commodity instances.}
{It remains open whether a similar approach as taken here to show long term behaviour can} be developed in the multi-commodity setting. 
Progress here would be very exciting.

	\ifdefined\journal
	\paragraph{Acknowledgments.}
        We thank the anonymous referees for their detailed and constructive feedback, which led to a number of improvements.

        (Further acknowledgements omitted until after peer review).
	\else
	\paragraph{Acknowledgments.}
	We are grateful to Jos\'e Correa, Andr\'es Cristi, Dario Frascaria, Marcus Kaiser, Tim Oosterwijk and Philipp Warode for many interesting discussions on various topics related to equilibria in flows over time.
	The first author thanks Dario Frascaria in addition for useful discussions some years prior to this work on the topic of long-term behavior without a bounded inflow assumption.
	 Finally, we thank the anonymous referees for their detailed and constructive feedback, which led to a number of improvements.
	\fi
	
	\bibliographystyle{alpha}
	\bibliography{literature}
\end{document}